\documentclass[a4paper,UKenglish,cleveref,autoref,thm-restate]{lipics-v2021}
\pdfoutput=1 %uncomment to ensure pdflatex processing (mandatatory e.g. to submit to arXiv)
\hideLIPIcs  %uncomment to remove references to LIPIcs series (logo, DOI, ...), e.g. when preparing a pre-final version to be uploaded to arXiv or another public repository

\usepackage{mathtools}
\usepackage{tikz} 
\usepackage{noindentafter}

\NoIndentAfterEnv{definition}
\NoIndentAfterEnv{example}
\NoIndentAfterEnv{lemma}
\NoIndentAfterEnv{claim}
\NoIndentAfterEnv{theorem}
\NoIndentAfterEnv{proof}
\NoIndentAfterEnv{claimproof}
\NoIndentAfterEnv{restatable}

\theoremstyle{claimstyle}
\newtheorem{subclaim}[theorem]{Subclaim}

\newcommand{\subsubparagraph}[1]{\textcolor{lipicsGray}{\textsf{\textbf{{#1}}}}}

\def\natnum{{\mathbb N}}

\newcommand{\lam}{\lambda} 
\newcommand{\sm}{\setminus}

\newcommand{\bigO}{\mathcal{O}}
 
\newcommand{\bbbn}{\mathbb{N}}

\newcommand{\bbbr}{\mathbb{R}}

\hypersetup{colorlinks   = true,
            citecolor    = black,
            urlcolor     = violet,
            linkcolor    = black}

\title{Quasi-Isometric Reductions Between Infinite Strings} %TODO Please add

\author{Karen Frilya Celine}{School of Computing, National University of Singapore, Singapore 117417}{karen.celine@u.nus.edu}{https://orcid.org/0000-0002-7078-5582}{} %TODO mandatory, please use full name; only 1 author per \author macro; first two parameters are mandatory, other parameters can be empty. Please provide at least the name of the affiliation and the country. The full address is optional. Use additional curly braces to indicate the correct name splitting when the last name consists of multiple name parts.

\author{Ziyuan Gao}{School of Science and Technology, Singapore University of Social Sciences, Singapore 599494 \and Kaplan Higher Education Academy, Singapore 228095}{ziyuan84@yahoo.com}{}{}

\author{Sanjay Jain}{School of Computing, National University of Singapore, Singapore 117417}{sanjay@comp.nus.edu.sg}{https://orcid.org/0000-0001-6798-8330}{}

\author{Ryan Lou}{School of Computing, National University of Singapore, Singapore 117417}{ryan.lou@u.nus.edu}{}{}

\author{Frank Stephan}{Department of Mathematics, National University of Singapore, Singapore 119076 \and School of Computing, National University of Singapore, Singapore 117417}{fstephan@comp.nus.edu.sg}{https://orcid.org/0000-0001-9152-1706}{}

\author{Guohua Wu}{School of Physical and Mathematical Sciences, Nanyang Technological University, Singapore 637371}{guohua@ntu.edu.sg}{}{}

\authorrunning{K. F. Celine, Z. Gao, S. Jain, R. Lou, F. Stephan, and G. Wu} %TODO mandatory. First: Use abbreviated first/middle names. Second (only in severe cases): Use first author plus 'et al.'

\Copyright{Karen Frilya Celine, Ziyuan Gao, Sanjay Jain, Ryan Lou, and Frank Stephan} %TODO mandatory, please use full first names. LIPIcs license is "CC-BY";  http://creativecommons.org/licenses/by/3.0/

 %TODO mandatory: Please choose ACM 2012 classifications from https://dl.acm.org/ccs/ccs_flat.cfm 

 %TODO mandatory; please add comma-separated list of keywords

\category{} %optional, e.g. invited paper

\relatedversion{A shorter version of this paper has been accepted to the 49th International Symposium on Mathematical Foundations of Computer Science (MFCS 2024).}
\funding{S.~Jain and F.~Stephan were supported by Singapore Ministry of Education (MOE) AcRF Tier 2 grant MOE-000538-00. Additionally, S.~Jain was supported by NUS grant E-252-00-0021-01. G.~Wu was partially supported by MOE Tier 1 grants RG111/19 (S) and RG102/23.}%optional, to capture a funding statement, which applies to all authors. Please enter author specific funding statements as fifth argument of the \author macro.

\acknowledgements{We thank the anonymous referees of MFCS 2024 as well as JCSS for their helpful comments which improved the presentation of the paper.} %optional

\nolinenumbers %uncomment to disable line numbering

\EventEditors{Rastislav Kr\'{a}lovi\v{c} and Anton\'{i}n Ku\v{c}era}
\EventNoEds{2}
\EventLongTitle{49th International Symposium on Mathematical Foundations of Computer Science (MFCS 2024)}
\EventShortTitle{MFCS 2024}
\EventAcronym{MFCS}
\EventYear{2024}
\EventDate{August 26--30, 2024}
\EventLocation{Bratislava, Slovakia}
\EventLogo{}
\SeriesVolume{306}
\ArticleNo{29}
\begin{document}

\maketitle

%TODO mandatory: add short abstract of the document
\begin{abstract}
This paper studies the recursion- and automata-theoretic aspects of large-scale geometries
of infinite strings, a subject initiated by Khoussainov and Takisaka (2017).
We first investigate several notions of quasi-isometric reductions between recursive infinite strings and prove various results on the equivalence classes of
such reductions. 
The main result is the construction of two infinite recursive strings $\alpha$
and $\beta$ such that $\alpha$ is strictly quasi-isometrically reducible
to $\beta$, but the reduction cannot be made recursive.
This answers an open problem posed by Khoussainov and Takisaka.
Furthermore, we also study automatic quasi-isometric reductions between automatic structures,
and show that automatic quasi-isometry may be separable from general quasi-isometry
depending on the growth of the automatic domain.
\end{abstract}

\section{Introduction}

Quasi-isometry is an important concept in geometric group theory
that has been used to solve problems in group theory.
Loosely speaking, two metric spaces are said to be quasi-isometric 
iff there is a mapping (called a {\em quasi-isometry}) from one 
metric space to the other that preserves the distance between any 
two points in the first metric space up to some multiplicative and
additive constants. Thus, for example, while the Euclidean plane is 
not isometric to $\mathbb{R}^2$ equipped with the taxicab distance,
the two spaces are quasi-isometric to each other since the Euclidean
distance between any two points does not differ from the
taxicab distance between them up to a multiplicative factor of $\sqrt{2}$.
The study of group properties---where groups are represented by
their Cayley graphs---that are invariant under
quasi-isometries is quite a prominent theme in geometric group theory;
examples of such group properties include hyperbolicity and growth rate
\cite{Gromov81}. 
 
Khoussainov and Takisaka \cite{Khoussainov17} introduced a refined
notion of quasi-isometry for infinite strings called 
{\em colour-preserving quasi-isometry}, thus enabling the
study of global patterns on strings and linking the study
of large-scale geometries with automata theory, recursion theory and model theory.
Here, an infinite string $\alpha \in \Sigma^\omega$ is represented by
the metric space $\bbbn$
where each number $i$ has a specific colour $\sigma \in \Sigma$ 
depending on the $i$-th symbol in the string $\alpha$.
This paper builds on the results from 
Khoussainov and Takisaka \cite{Khoussainov17} further,
by studying, in particular, {\em recursive} quasi-isometry between 
{\em recursive} infinite strings.

In addition to Khoussainov and Takisaka's paper \cite{Khoussainov17}
which studied quasi-isometry between infinite strings specifically,
quasi-isometries between hyperbolic metric spaces in general---an example of which
is an infinite string when viewed
as a coloured metric space---are well-studied in geometric
group theory. Isometries between computable metric spaces have also 
been studied by Melnikov \cite{Melnikov13}.

Among the various questions investigated by Khoussainov and Takisaka
was the computational complexity of the {\em quasi-isometry 
problem}: given any two infinite strings $\alpha$ and $\beta$, is
there a quasi-isometry from $\alpha$ to $\beta$? They found that
for any two quasi-isometric strings, a quasi-isometry 
that is recursive in the halting problem relative to $\alpha$ and $\beta$ 
always exists between them, and that
the quasi-isometry problem between any two recursive strings 
is $\Sigma^0_2$-complete \cite{Khoussainov22}\footnote{Note that 
\cite{Khoussainov22} is the journal version of \cite{Khoussainov17},
containing some corrections from the earlier paper.}.
In comparison, the corresponding problem for isometry with 
respect to recursive strings is $\Pi^0_1$-complete \cite{Melnikov13}.
Khoussainov and Takisaka also had the following open problem
which was mentioned in many talks and discussions:
{\em if a quasi-isometric reduction from $\alpha$ to $\beta$ exists,
does there always exist a recursive quasi-isometric reduction?}
This is a very natural question for computer science,
specifically for computability theory, since
it seeks to understand how complex such a reduction is.
We answered this question in the negative, that is,
there are cases where the reduction exists
but cannot be made recursive.
The fourth author's bachelor thesis \cite{Lou19} which contains this result
was cited by Khoussainov and Takisaka in the journal version \cite{Khoussainov22}
of their paper \cite{Khoussainov17}.

To complete the picture, 
the present work examines, in more detail, the recursion-theoretic 
aspects of quasi-isometries between infinite strings. 
We study various natural restrictions on quasi-isometric reductions between
strings: first, {\em many-one reductions}, where the quasi-isometric reduction 
is required to be {\em recursive} and {\em many-one}; second,
{\em one-one reductions}, which are injective many-one reductions;
third, {\em permutation quasi-isometric reductions}, which
are surjective one-one reductions. 

The main subjects of this work are the structural properties of the equivalence classes induced by the different types of reductions and the relationships between
these reductions. In accordance with recursion-theoretic terminology, 
we call an equivalence class induced by a reduction type a {\em degree} of
that reduction type. We show, for example, that within 
each many-one quasi-isometry degree, any pair of strings has a common upper 
bound as well as a common lower bound with respect to one-one reductions. 
Furthermore, there are two strings for which their many-one quasi-isometry
degrees have a unique least common upper bound. The main result is the 
separation of quasi-isometry from {\em recursive} quasi-isometry,
that is, we construct two recursive strings such that one is
quasi-isometrically reducible to the other but no recursive many-one
quasi-isometry exists between them. 
This main result answers the above-mentioned open problem
posed by Khoussainov and Takisaka.

In addition, we also investigate the automata-theoretic aspects of
quasi-isometries.
In this case, the question is whether automatic quasi-isometry can be separated
from general quasi-isometry.
Our results show that the answer depends on the growth of the domain.
In linear domains, quasi-isometry is equivalent to automatic quasi-isometry.
On the other hand, in superlinear but polynomial domains,
one can always find automatic colourings $\alpha$ and $\beta$ such that
$\alpha$ is quasi-isometrically reducible to $\beta$ but not automatically.
While it is known that quasi-isometry can be separated
from its automatic counterparts in some exponential domains,
it is unknown whether such separation is always possible
in all domains with exponential growth.

\section{Notation}

Any unexplained recursion-theoretic notation may be found in
\cite{Odifreddi89,Rogers67,Soare87}. The set of positive integers
will be denoted by $\bbbn$; $\bbbn \cup \{0\}$ will be denoted by
$\bbbn_0$.
The finite set $\Sigma$ will denote the alphabet used.
We assume knowledge of elementary computability
theory over different size alphabets~\cite{Calude02}.
An infinite string $\alpha \in \Sigma^\omega$ can also be
viewed as a $\Sigma$-valued function defined on $\bbbn$.  
The length of an interval $I$ is denoted by $|I|$.
For $\alpha_i \in \Sigma^*$ and $i \in \bbbn$, 
we write $(\alpha_i)_{i=1}^{\infty}$ to denote $\alpha_1\alpha_2\cdots$,
a possibly infinite string.

\section{Coloured Metric Spaces and Infinite Strings}
\label{section:coloured.metric.space}

\begin{definition}[Coloured Metric Spaces, \cite{Khoussainov17}]
\label{defn:colouredmetricspace}
A {\em coloured metric space $(M;d_M,Cl)$} consists of
the {\em underlying metric space} $(M;d_M)$ with metric $d_M$ 
and the colour function $Cl: M \to \Sigma$,
where $\Sigma$ is a finite set of colours called an {\em alphabet}.
We say that {\em $m \in M$ has colour $\sigma \in \Sigma$} if $\sigma = Cl(m)$.
\end{definition}

% Note:
% The original definition allows each $m \in M$ to have more than one colours.
% However, for infinite strings, each position $n \in \bbbn$ can only have one colour.
% Due to limited space, we use this simpler definition
% to avoid the need to explain this simplification.

\begin{definition}
[Quasi-isometries Between Coloured Metric Spaces, \cite{Khoussainov17}]
\label{defn:quasiisometry}

For any $A \geq 1$ and $B \geq 0$, an {\em $(A,B)$-quasi-isometry} from a metric space $\mathcal{M}_1 = (M_1;d_1)$ to a metric space 
$\mathcal{M}_2 = (M_2;d_2)$ is a function $f:M_1 \to M_2$ 
such that for all $x,y \in M_1$, 
$\frac{1}{A}\cdot d_1(x,y) - B \leq d_2(f(x),f(y)) \leq A\cdot d_1(x,y) + B$, 
and for all $y \in M_2$, there exists an $x \in M_1$ such that
$d_2(f(x),y) \leq A$.   

Given two coloured metric spaces $\mathcal{M}_1 = (M_1;d_1,Cl_1)$
and $\mathcal{M}_2 = (M_2;d_2,\linebreak[2]Cl_2)$, a function $f:M_1 \to M_2$
is a {\em quasi-isometric reduction} from $\mathcal{M}_1$ to
$\mathcal{M}_2$ iff for some $A \geq 1$ and $B \geq 0$, 
$f$ is an $(A,B)$-quasi-isometry from $(M_1;d_1)$ to $(M_2;d_2)$ 
and $f$ is {\em colour-preserving}, that is, for all 
$x \in M_1$, $Cl_1(x) = Cl_2(f(x))$. 
\end{definition}

An infinite string $\alpha$ can then be seen as a coloured metric space
$(\bbbn;d,\alpha)$, where $d$ is the metric on $\bbbn$ defined by
$d(i,j) = |i-j|$ and $\alpha:\bbbn \to \Sigma$ is the colour function.
For any two infinite strings $\alpha$ and $\beta$, we write 
$\alpha \leq_{qi} \beta$ to mean that there is a quasi-isometric reduction
from $\alpha$ to $\beta$. The relation $\leq_{qi}$ is a preorder on
$\Sigma^{\omega}$. For any pair of distinct letters $a_1, a_2 \in \Sigma$,
$a_1^{\omega}$ and $a_2^{\omega}$ are incomparable
with respect to $\leq_{qi}$, so this relation is not total.

The following proposition gives a useful
simplification of the definition of quasi-isometry in the context of
infinite strings.

\begin{proposition}\label{prop:quasi.iso.alt.def}
Given two infinite strings $\alpha$ and $\beta$, let 
$f:\bbbn\to\bbbn$ be a colour-preserving function.
Then, $f$ is a quasi-isometric reduction from $\alpha$ to $\beta$ iff
there exists a constant $C \geq 1$ such that for all $x,y$ in the domain
of $\alpha$, the following conditions hold:
\begin{enumerate}[(a)]
\item $d(f(x),f(x+1)) \leq C$;
\item $x + C < y \Rightarrow f(x) < f(y)$.
\end{enumerate}
\end{proposition}

\begin{proof}
First, suppose that $f:\bbbn\to\bbbn$ is a colour-preserving 
quasi-isometric reduction from $\alpha$ to $\beta$.
We show that there exists a constant $C \geq 1$ for which Conditions (a) and 
(b) hold for any $x,y \in \natnum$. By the definition of a quasi-isometric reduction, there exist constants $A \geq 1$ and $B \geq 0$ such that 
\begin{equation}\label{eqn:quasi.iso.ineq}
\frac{1}{A}\cdot d(x,y) - B \leq d(f(x),f(y)) \leq A\cdot d(x,y) + B.
\end{equation}
We first derive, for each of the two conditions, a choice of $C$ satisfying it. 
\begin{enumerate}[(i)]
\item Plugging $y = x+1$ into the upper bound in (\ref{eqn:quasi.iso.ineq}) 
yields $d(f(x),f(x+1)) \leq A + B$.
\item Assume for the sake of a contradiction that for all $C \geq 1$,
there are $x \in \bbbn$ and $C' > C$ such that $f(x+C') \leq f(x)$.
We show that if $C$ is chosen so that $A+B \leq \frac{1}{A}\cdot C - B$, 
then the existence of some $C' > C$ with $f(x+C') \leq f(x)$ would lead
to a contradiction. Fix such a $C$, and suppose there were indeed some 
$C'$ with $C' > C \geq 1$ and 
\begin{equation}\label{eqn:quasi.iso.lb1}
f(x+C') \leq f(x)\,.
\end{equation} 
Then,
\[
\begin{aligned}
f(x+C'+1) - f(x+C') &\leq d(f(x+C'+1),f(x+C')) \\
&\leq A+B~~(\text{by statement (i)})\\
&\leq \frac{1}{A}\cdot C - B~~(\text{by the choice of $C$}) \\
&< \frac{1}{A}\cdot C' - B~~(\text{since $C' > C$}) \\
&\leq f(x) - f(x+C')~~(\text{by (\ref{eqn:quasi.iso.ineq}) and (\ref{eqn:quasi.iso.lb1})})\,,
\end{aligned}
\] 
giving $f(x+C'+1) < f(x)$. One can repeat the preceding argument inductively,
yielding the inequality $f(x+C'+k+1) - f(x+C'+k) < f(x) - f(x+C'+k)$, or
equivalently $f(x+C'+k+1) < f(x)$, for each $k \geq 0$. But this is impossible since 
$f(x)$ is finite and $d(f(x+C'+k+1), f(x+C'+k'+1)) > 0$ whenever $|k-k'|$ is
sufficiently large. 
\end{enumerate} 
It follows from (i) and (ii) that Conditions (a) and (b) are satisfied for 
$C = A\cdot(A+2B)$.

For a proof of the converse direction, 
fix a $C$ satisfying Conditions (a) and (b). 
Suppose $x \in \bbbn$. Then, by Condition (a), $d(f(x),f(x+1)) \leq C$.
Inductively, assume that $d(f(x),f(x+n)) \leq n\cdot C$.
Then, by the inductive hypothesis and Condition (a), 
$d(f(x),f(x+n+1)) \leq d(f(x),f(x+n)) + d(f(x+n),f(x+n+1)) 
\leq n\cdot C + C = (n+1)\cdot C$ where the first inequality follows from 
the triangle inequality. Consequently, for all $x,y \in \bbbn$,
$d(f(x),f(y)) \leq d(x,y)\cdot C$.
    
Next, we establish a lower bound for $d(f(x),f(y))$.
Without loss of generality, assume $x < y$.
Write $y = x+i(C+1)+j$ for some $i \in \bbbn_0$
and $0 \leq j \leq C$. By a simple induction, one can show that
$f(x+i(C+1)) \geq f(x)+i$ and thus $d(f(x),f(x+i(C+1)) \geq i$.
Furthermore, $d(f(x+i(C+1)),f(y)) \leq C^2$. Thus,
$d(f(x),f(y)) \geq i-C^2$ and $i \geq d(x,y)/(C+1)-1$.
It follows that $d(f(x),f(y)) \geq d(x,y)/(C+1)-1-C^2$.

Thus, one can select $A=C+1$ and $B=C^2+1$ to establish
the required bounds for the quasi-isometric mapping.
\end{proof}

% \begin{remark}\label{rem:quasiisodef}
% In the first part of the proof of Proposition \ref{prop:quasi.iso.alt.def}, 
% the derivation of Conditions (a) and (b) only used the first condition in the
% definition of $(A,B)$-quasi-isometry, namely, that there are constants
% $A \geq 1$ and $B \geq 0$ such that 
% $\frac{1}{A}\cdot d(x,y) - B \leq d(f(x),f(y)) \leq A\cdot d(x,y) + B$.
% Thus, in the setting of infinite strings, the second condition -- that 
% for any $y \in \bbbn$, there exists some $x \in \bbbn$ with 
% $d(y,f(x)) \leq A$ -- is actually implied by the first condition
% (though possibly with a different choice of $A$).
% \end{remark}

By Proposition \ref{prop:quasi.iso.alt.def}, we can now redefine 
quasi-isometric reduction in terms of one constant $C$,
instead of two constants $A$ and $B$ as in Definition \ref{defn:quasiisometry},
reducing the number of constants by 1.

\begin{definition}\label{defn:C.reduction}
Suppose $C \geq 1$. Given infinite strings $\alpha$ and $\beta$, a 
{\em $C$-quasi-isometry} from $\alpha$ to $\beta$ is 
a colour-preserving function $f:\bbbn\to\bbbn$
such that for all $x,y$ in the domain
of $\alpha$,
\begin{enumerate}[(a)]
\item $f(1)\leq C$ and $f(x)-C \leq f(x+1) \leq f(x)+C$;
\item $x + C < y \Rightarrow f(x) < f(y)$.
\end{enumerate}
For the rest of the paper, we shall use ``Condition (a)'' and ``Condition (b)''
to refer to the above conditions respectively,
without necessarily mentioning the definition number.
\end{definition}

A useful property of a $C$-quasi-isometry $f$ from $\alpha$ to $\beta$ 
is that any position of $\beta$ has at most $C+1$ pre-images
under $f$.

\begin{lemma}
[{\cite[Corollary II.4]{Khoussainov17}}]
\label{lem:collision}
Given two infinite strings $\alpha$ and $\beta$, suppose that 
$f$ is a $C$-quasi-isometry from $\alpha$ to $\beta$.
Then, for all $y \in \bbbn$,
$|f^{-1}(y)| \leq C+1$.
\end{lemma} 

\begin{proof}
If $y \in \bbbn$ is not in the range of $f$, then $|f^{-1}(y)| = 0$.
So assume that $f^{-1}(y)$ is not empty. Set  $x_0 = \min\{x: f(x) = y\}$.
Now for all $x \in f^{-1}(y)$, we must have $x_0 + C \geq x \geq x_0$.
The first inequality follows from Condition (b) since $f(x_0) = y = f(x)$.
The second inequality follows from the choice of $x_0$.
Hence, $f^{-1}(y) \subseteq \{x_0, \ldots, x_0 + C\}$
and so $|f^{-1}(y)| \leq C + 1$.
\end{proof}

It was proven earlier that for any infinite strings $\alpha, \beta$
and any $C$-quasi-isometry $f$ from $\alpha$ to $\beta$, there is a constant 
$D$ such that each position of $\beta$ is at most $D$ positions away 
from some image of $f$. The next lemma states that each position of
$\beta$ in the range of $f$ is at most $C$ positions away from a 
different image of $f$.

\begin{restatable}{lemma}{imagepointdensity}
\label{lem:image.point.density}
Let $\alpha$ and $\beta$ be infinite strings and
let $f$ be a $C$-quasi-isometry from $\alpha$ to $\beta$.
Then, $\min\{f(x): x \in \bbbn\} \leq C$ and for each $y \in \bbbn$,
$\min\{f(x) > f(y): x \in \bbbn\} \leq f(y) + C$.
Hence, for each $z \in \bbbn$, there is some
$x \in \bbbn$ such that $d(f(x),z) \leq C$.
\end{restatable}

\begin{proof}
By definition, $\min\{f(x): x \in \bbbn\} \leq f(1) \leq C$.
For each $y \in \bbbn$, take the smallest $y'$ such that $f(y') > f(y)$.
Such $y'$ must exist by Condition (b).
If $y' = 1$, then $f(y') \leq C \leq f(y) + C$.
Otherwise, $f(y' - 1) \leq f(y)$.
By Condition (a), $f(y') \leq f(y' - 1) + C \leq f(y) + C$.
Hence, $\min\{f(x) > f(y): x \in \bbbn\} \leq f(y') \leq f(y) + C$.
\end{proof}

\begin{restatable}{corollary}{naryseq}\label{cor:nary.seq}
Let $\Sigma = \{a_1,\ldots,a_{l}\}$ and let $\alpha,\beta$ be two 
infinite strings. Let $f$ be a $C$-quasi-isometry from $\alpha$ to $\beta$.
Suppose that there is a positive integer $K$ such that there is at  least
one occurrence of $a_i$ in any interval of positions of $\alpha$ of length $K$.
Then, there is at least one occurrence of $a_i$ in any interval of 
positions of $\beta$ of length $KC$. 
\end{restatable}

\begin{proof}
Consider an interval $[y+1, y+KC]$ of positions of $\beta$ of length $KC$.
By Condition (b), there is a least $x$ such that $f(x) > y$ and $\alpha(x) = a_i$.
We show that $f(x)$ lies in the interval $[y+1, y+KC]$ and so is a position of some occurrence of $a_i$ in $\beta$.
Suppose that there is no occurrence of $a_i$ before position $x$.
Then, $x \leq K$, since there must be at least one occurrence of $a_i$
in the first $K$ positions of $\alpha$.
Hence, by Condition (a), $f(x) \leq KC \leq y + KC$.
Otherwise, let $x'$ be the position of the last occurrence of $a_i$
before position $x$.
By definition of $x$, $f(x') \leq y$, and by the assumption, $x - x' \leq K$.
Then, by Condition (a), $f(x) \leq f(x') + KC \leq y + KC$.
In both cases, we have $y + 1 \leq f(x) \leq y + KC$.
\end{proof}

A quasi-isometry $f$ can fail to be order-preserving in that 
there are pairs $x,y \in \bbbn$ with $x < y$ and $f(x) > f(y)$.
Nonetheless, as Khoussainov and Takisaka noted 
\cite[Lemma II.2]{Khoussainov17}, every quasi-isometry enforces a uniform upper bound on the size of a {\em cross-over}---the difference $f(x) - f(y)$
for such a pair $x,y \in \bbbn$. This may be proven using the alternative
definition of a quasi-isometry given in 
Definition \ref{defn:C.reduction}.

\begin{lemma}
[Small Cross-Over Lemma, \protect{\cite[Lemma II.2]{Khoussainov17}}]
\label{lem:small.crossover}
Given two infinite strings $\alpha$ and $\beta$, suppose that 
$f$ is a $C$-quasi-isometry from $\alpha$ to $\beta$.
Then, for all $n,m \in \bbbn$ with $n < m$, we have $f(n)-f(m) \leq C^2$.
\end{lemma}

\begin{proof}
Consider any $n,m \in \bbbn$ with $n < m$.
If $n + C < m$, then by Condition (b), $f(n) - f(m) < 0$. 
If $n + C \geq m$, then $d(m,n) = m-n \leq C$, so by applying 
Condition (a) $m-n$ times, one has
$f(n) - f(m) \leq d(f(m),f(n)) \leq (m-n)\cdot C \leq C^2$.
Thus, $f(n)-f(m) \leq C^2$ holds for any choices of $m,n \in \bbbn$
with $n < m$.  
\end{proof}

\section{Recursive Quasi-Isometric Reductions}\label{section:rec.quasi.iso.red}

Khoussainov and Takisaka \cite{Khoussainov17} investigated
the structure of the partial-order $\Sigma^{\omega}_{qi}$ of the 
quasi-isometry degrees over an alphabet $\Sigma = \{a_1,\ldots,a_l\}$. 
They proved that $\Sigma^{\omega}_{qi}$ has a greatest element, 
namely the degree of $(a_1\cdots a_n)^{\omega}$, and that 
$\Sigma^{\omega}_{qi}$ contains uncountably many minimal elements. 
Furthermore, they showed that $\Sigma^{\omega}_{qi}$ includes a chain of 
the type of the integers, and that it includes an antichain. In connection with
computability theory, in particular with the arithmetical hierarchy, 
they established that the quasi-isometry relation on recursive infinite 
strings is $\Sigma^0_2$-complete \cite{Khoussainov22}. In this section,
we continue research into the recursion-theoretic aspects of quasi-isometries
on infinite strings. 
We consider the notions of many-one and one-one recursive reducibilities
first introduced by Post \cite{Post44} as relations between recursive functions,
and apply them to quasi-isometric reductions.
We also define a third type of quasi-isometric reducibility---permutation reducibility---which is bijective.
We then prove a variety of results on the degrees of such reductions.

\begin{definition}[Many-One Reducibility]\label{defn:many.one.quasi.red}
A string $\alpha$ is {\em many-one reducible}, or {\em mqi-reducible}, 
to a string $\beta$ iff there exists a quasi-isometric reduction $f$ 
from $\alpha$ to $\beta$ such that $f$ is {\em recursive}.
We call such an $f$ a {\em many-one quasi-isometry (or mqi-reduction)}, 
and write $\alpha \leq_{mqi} \beta$ to mean that $\alpha$ is many-one 
reducible to $\beta$; if, in addition, $f$ is a $C$-quasi-isometry, then
we call $f$ a {\em $C$-many-one quasi-isometry} (or {\em $C$-mqi-reduction}). 
We write $\alpha <_{mqi} \beta$ to mean that $\alpha \leq_{mqi} \beta$
and $\beta \not\leq_{mqi} \alpha$.
\end{definition}

\begin{definition}[One-One Reducibility]\label{defn:one.one.quasi.red}
A string $\alpha$ is {\em one-one reducible}, or {\em 1qi-reducible},
to a string $\beta$ iff there exists a many-one quasi-isometry $f$
from $\alpha$ to $\beta$ such that $f$ is one-one.
We call such an $f$ a {\em one-one quasi-isometry (or 1qi-reduction)}, 
and write $\alpha \leq_{1qi} \beta$ to mean that $\alpha$ is one-one reducible
to $\beta$; if, in addition, $f$ is a $C$-quasi-isometry, then we call
$f$ a {\em $C$-one-one quasi-isometry} (or {\em $C$-1qi-reduction}). 
We write $\alpha <_{1qi} \beta$ to mean that $\alpha \leq_{1qi} \beta$
and $\beta \not\leq_{1qi} \alpha$.
\end{definition}

\begin{definition}[Permutation Reducibility]\label{defn:permute.quasi.red}
A string $\alpha$ is {\em permutation reducible}, or {\em pqi-reducible},
to a string $\beta$ iff there exists a one-one quasi-isometry $f$
from $\alpha$ to $\beta$ such that $f$ is surjective.
We call such an $f$ a {\em permutation quasi-isometry (or pqi-reduction)}, 
and write $\alpha \leq_{pqi} \beta$ to mean that $\alpha$ is permutation reducible
to $\beta$; if, in addition, $f$ is a $C$-quasi-isometry, then we call
$f$ a {\em $C$-permutation quasi-isometry} (or {\em $C$-pqi-reduction}). 
We write $\alpha <_{pqi} \beta$ to mean that $\alpha \leq_{pqi} \beta$
and $\beta \not\leq_{pqi} \alpha$.
\end{definition}

Given an alphabet $\Sigma$, the relations $\leq_{mqi}$, $\leq_{1qi}$, 
$\leq_{pqi}$ and $\leq_{qi}$ are preorders on the class of infinite strings over 
$\Sigma$. Let $\equiv_{mqi}$ be the 
relation on $\Sigma^{\omega}$ such that 
$\alpha \equiv_{mqi} \beta$ iff $\alpha \leq_{mqi} \beta$ and 
$\beta \leq_{mqi} \alpha$.
Then, $\equiv_{mqi}$ is an equivalence relation on $\Sigma^{\omega}$.  
We call an equivalence class on $\Sigma^{\omega}$ 
induced by $\equiv_{mqi}$ a {\em many-one quasi-isometry degree} 
(or {\em mqi-degree}), and denote
the mqi-degree of an infinite string $\alpha$ by $[\alpha]_{mqi}$. 
Analogous definitions apply to $\equiv_{1qi}$, $[\alpha]_{1qi}$,
$\equiv_{pqi}$, $[\alpha]_{pqi}$, $\equiv_{qi}$ and $[\alpha]_{qi}$.

We denote the partial orders induced by $\leq_{pqi},
\leq_{1qi},\leq_{mqi}$ and $\leq_{qi}$ 
on the pqi-degrees, 1qi-degrees,
mqi-degrees and qi-degrees by $\Sigma^{\omega}_{pqi},
\Sigma^{\omega}_{1qi},\Sigma^{\omega}_{mqi}$ and 
$\Sigma^{\omega}_{qi}$ respectively. 
 
By definition, $\Sigma^{\omega}_{pqi}$ is a refinement
of $\Sigma^{\omega}_{1qi}$ in the sense that 
for all infinite strings $\alpha$ and $\beta$,
$[\alpha]_{pqi} \leq_{pqi} [\beta]_{pqi} \Rightarrow
[\alpha]_{1qi} \leq_{1qi} [\beta]_{1qi}$.
In a similar manner, $\Sigma^{\omega}_{1qi}$ is a refinement
of $\Sigma^{\omega}_{mqi}$, which is in turn a refinement
of $\Sigma^{\omega}_{qi}$. The first subsection
deals with the mqi-degrees, starting with the inner structure of 
each mqi-degree.

\subsection{Structure of the mqi-Degrees}\label{subsection:mqi.degrees} 

Fix any two distinct infinite strings $\beta$ and $\gamma$ 
belonging to $[\alpha]_{mqi}$. It can be shown that $\beta$
and $\gamma$ have a common upper bound as well as a common
lower bound in $[\alpha]_{mqi}$ such that these bounds are witnessed
by 1qi-reductions.

\begin{proposition}\label{prop:mqi.ub}
For any two distinct infinite strings $\beta,\gamma \in [\alpha]_{mqi}$,
there exists a $\delta \in [\alpha]_{mqi}$ such that
$\beta \leq_{1qi} \delta$ and $\gamma \leq_{1qi} \delta$.
\end{proposition}

\begin{proof}
Let $f$ be a $C$-mqi-reduction from $\beta$ to $\gamma$.
Let $\delta$ be the infinite string obtained from $\gamma$ by 
repeating $C+1$ times each letter of $\gamma$. Then,
$\gamma \leq_{1qi} \delta$ via a $(C+1)$-1qi-reduction $g$ 
defined by $g(n) = (n-1)\cdot(C+1)+1$ for each $n \in \bbbn$.
Furthermore, $\delta \leq_{mqi} \gamma$ via a $C$-mqi-reduction
$g'$ defined by $g'(n) = \lceil \frac{n}{C+1} \rceil$. 
Thus, $\delta \in [\alpha]_{mqi}$.

Next, one constructs a $(C^2+2C)$-1qi-reduction $f'$ from $\beta$
to $\delta$ using the function $f$. For each $y$ in the range
of $f$, map the pre-image of $y$ under $f$, which by
Lemma \ref{lem:collision} has at most $C+1$ elements,
to the set of positions of $\delta$ corresponding to the
$C+1$ copies of the letter at position $y$. Formally, define   
\[
f'(n) = \left\{\begin{array}{ll}
g(f(n))\,, & \mbox{if $f(n) \neq f(n')$ for all $n' < n$;} \\
g(f(n)) + C'\,, & \mbox{otherwise;
where $1 \leq C' < C+1$ is minimum such that} \\
~ & \mbox{$g(f(n)) + C' \neq f'(n')$ for all $n' < n$.}
\end{array}\right.
\]
We verify that $f'$ is an injective $(C^2+2C)$-quasi-isometry.
Injectiveness follows from the definition of $f'$:
in the first case, the injectiveness of $g$ 
ensures that $f'(x) \neq f'(x')$ for all $x' < x$;
in the second case, it is directly enforced that $f'(x) \neq
f(x')$ for all $x' < x$.
Since $f$ is a $C$-reduction, $x + C < y \Rightarrow
f(x) < f(y) \Rightarrow g(f(x)) < g(f(y)) \Rightarrow
f'(x) < f'(y)$, and so $f'$ satisfies Condition (b) with 
constant $C$. 
Now we show that $f'$ satisfies Condition (a) with constant 
$C^2+2C$.
By Condition (a),
$d(f(x),f(x+1)) \leq C$. Without loss of generality, assume that
$f(x) \leq f(x+1)$. By the definition of $f'$, 
$f'(x) \geq g(f(x))$ and $f'(x+1) \leq g(f(x+1)) + C$.
Since $f(x) \leq f(x+1)$, it follows that $f'(x) \leq f'(x+1)$
and so 
\[
\begin{aligned}
d(f'(x+1),f'(x)) &\leq g(f(x+1)) + C - g(f(x)) \\
&= (C+1)\cdot (f(x+1)-1) + 1 + C - (C+1)\cdot (f(x)-1) - 1 \\
&= (C+1)\cdot(f(x+1)-f(x)) + C \\
&\leq C\cdot(C+1) + C \\
&= C^2 + 2C\,.
\end{aligned}
\]   
This completes the proof.
\end{proof}

Next, we prove a lower bound counterpart of Proposition \ref{prop:mqi.ub}.

\begin{proposition}\label{prop:mqi.lb}
For any two distinct infinite strings 
$\beta,\gamma \in [\alpha]_{mqi}$,
there exists a $\delta \in [\alpha]_{mqi}$ such that
$\delta \leq_{1qi} \beta$ and $\delta \leq_{1qi} \gamma$.
\end{proposition}

\begin{proof}
Suppose $\beta = \beta_1\beta_2\ldots$, where $\beta_i \in \Sigma$.
Let $f:\bbbn\to\bbbn$ be a $C$-mqi-reduction from $\beta$ to $\gamma$.
Now define $\delta = \beta_{i_1}\beta_{i_2}\ldots$, where
$i_k$ is the minimum index such that $i_k \neq i_l$ for all $l < k$ 
and for all $j < i_k$, $f(j) \neq f(i_k)$. 
By Condition (b),
the range of $f$ is infinite and thus each $i_k$ is well-defined. 
We verify that $\delta \leq_{1qi} \beta$ and $\delta \leq_{1qi} \gamma$.

Define $f'(n) = i_n$ for all $n \in \bbbn$.
We show that $f'$ is a 1qi-reduction from $\delta$ to $\beta$.
By the choice of the $i_n$'s, $f'(n) > f'(m)$ whenever $n > m$;
in particular, $f'$ is injective and Condition (b) holds for $f'$. 
Furthermore, given any $n$, by applying Condition (b) to $f$ and 
all $n' \leq n$, it follows that $f'(n+1) \leq f'(n)+C+1$.
Hence, $f'$ also satisfies Condition (a). 

Next, define a 1qi-reduction $f''$ from $\delta$ to $\gamma$ by
$f''(n) = f(i_n)$. The injectiveness of $f''$ follows from
the choice of the $i_n$'s (though $f''$ is not necessarily
strictly monotone increasing). Using the fact
that $i_{n+1} \leq i_n+C+1$, as well as applying Condition (a)
$i_{n+1} - i_n$ times, $d(f''(n+1),f''(n)) = d(f(i_{n+1}),f(i_n))
\leq C\cdot d(i_{n+1},i_n) \leq C\cdot(C+1)$. Hence,
$f''$ satisfies Condition (a) with constant $C\cdot(C+1)$.       
Since the $i_n$'s are strictly increasing, 
$m+C < n \Rightarrow i_{m} + C < i_{n} \Rightarrow f(i_m) < f(i_n)$. 
Thus, $f''$ is a $C\cdot(C+1)$-1qi-reduction.   

Lastly, define a mqi-reduction $g$ from $\beta$ to $\delta$
by $g(n) = k$ where $k$ is the minimum integer with $f(n) = f(i_k)$.
As the $i_n$'s cover the whole range of $f$, $g$ is well-defined.
For any given $n$, suppose $g(n) = k_1$ and $g(n+1) = k_2$, so
that $f(n) = f(i_{k_1})$ and $f(n+1) = f(i_{k_2})$.
By Condition (b), 
$d(n,i_{k_1}) \leq C$ and $d(n+1,i_{k_2}) \leq C$,
and so 
\[
\begin{aligned}
d(g(n),g(n+1)) &= d(k_1,k_2) \\
&\leq d(i_{k_1},i_{k_2}) \\
&\leq d(n,i_{k_1}) + d(n,n+1) + d(n+1,i_{k_2}) \\ 
&\leq 2C+1\,.
\end{aligned}
\]   
Hence, $g$ satisfies Condition (a) with constant $2C+1$.
To verify that $g$ satisfies Condition (b) for some constant,
fix any $n$ and apply Condition (b) $C\cdot(C+1)$ times to $f$,
giving $f(n)+C\cdot(C+1) \leq f(n+C\cdot(C+1)^2)$.
Suppose $g(n) = i_{k_1}$ and $g(n+C\cdot(C+1)^2) = i_{k_2}$,
so that $f(n) = f(i_{k_1})$ and $f(n+C\cdot(C+1)^2) = f(i_{k_2})$.
Then, 
$d(f(i_{k_1}),f(i_{k_2})) = d(f(n),f(n+C\cdot(C+1)^2)) \geq C\cdot(C+1)$. 
So by applying Condition (a) $d(i_{k_1},i_{k_2})$ times to $f$, we have
$C\cdot d(i_{k_1},i_{k_2}) \geq d(f(i_{k_1}),f(i_{k_2})) 
\geq C\cdot(C+1)$. Dividing both sides of the inequality by $C$
yields $d(i_{k_1},i_{k_2}) \geq C+1$.
Applying the contrapositive of Condition (b) to $f$ then gives
$f(i_{k_2}) \geq f(i_{k_1}) \Rightarrow i_{k_2} + C \geq i_{k_1}$.
Since $d(i_{k_1},i_{k_2}) \geq C+1$, this implies that
$g(n+C\cdot(C+1)^2) \\= i_{k_2} > i_{k_1} = g(n)$.
Thus, $g$ satisfies Condition (b) with constant $C\cdot(C+1)^2-1$.
\end{proof}

\subsection{1qi-Degrees Within mqi-Degrees}

We now investigate the structural properties of 1qi-degrees within 
individual mqi-degrees. As will be seen shortly, these properties 
can vary quite a bit depending on the choice of the mqi-degree. 

\begin{proposition}\label{prop:union.asc.chain}
There exists an infinite string $\alpha$ such that $[\alpha]_{mqi}$
is the union of an infinite ascending chain of 1qi-degrees.
\end{proposition}

\begin{proof}
Let $\Sigma = \{0,1\}$ and let $\alpha = 10^{\omega}$.
Then, $[\alpha]_{mqi}$ consists of all infinite strings with a 
finite, positive number of occurrences of $1$.
Given any infinite string $\beta$ with $k \geq 1$ 
occurrences of $1$, $\beta$ is 1qi-equivalent to a string 
$\gamma$ in $[\alpha]_{mqi}$ iff 
$\gamma$ has exactly $k$ occurrences of $1$. 
If $1 \leq k < k'$, then each string $\beta \in [\alpha]_{mqi}$
with exactly $k$ occurrences of $1$ is 1qi-reducible
to any string $\beta' \in [\alpha]_{mqi}$ with exactly
$k'$ occurrences of $1$. Thus, $[\alpha]_{mqi}$ is the union
of an ascending chain $[\alpha]_{1qi} < [110^{\omega}]_{1qi}
< [1110^{\omega}]_{1qi} < \ldots$, where the $i$-th term of
this chain is $1^i0^{\omega}$.\end{proof}

\begin{proposition}\label{prop:union.inf.disjoint.chains}
There exists an infinite string $\alpha$ such that the poset of 1qi-degrees 
within $[\alpha]_{mqi}$ is isomorphic to $\bbbn^2$ with the componentwise 
ordering. That is, $[\alpha]_{mqi}$
is the union of infinitely many disjoint infinite ascending chains of
1qi-degrees such that every pair of these ascending chains has
incomparable elements. Also, $[\alpha]_{mqi}$ does not contain 
infinite anti-chains of 1qi-degrees. 
\end{proposition}

\begin{proof}
Let $\Sigma = \{0,1,2\}$ and let $\alpha = 120^{\omega}$.
Then, $[\alpha]_{mqi}$ consists of all infinite strings with
a finite, positive number of $1$'s and a finite, positive
number of $2$'s. Furthermore, $[\alpha]_{1qi}$ consists of 
all infinite strings with exactly one occurrence of $1$
and exactly one occurrence of $2$.

Based on the proof of Proposition \ref{prop:union.asc.chain},
$[\alpha]_{mqi}$ is the union, over all $k \geq 1$, 
of chains of the form 
$[12^k0^{\omega}]_{1qi} < [1^22^k0^{\omega}]_{1qi} < 
\ldots$, where the $i$-th term of each chain is $[1^i2^k0^{\omega}]_{1qi}$.
Given any two chains $\Gamma_{j} = \{[1^i2^j0^{\omega}]_{1qi}: i \in \bbbn\}$ 
and $\Gamma_{k} = \{[1^i2^k0^{\omega}]_{1qi}: i \in \bbbn\}$,
where $j < k$, the classes $[1^{2}2^j0^{\omega}]_{1qi} \in \Gamma_j$ 
and $[12^k0^{\omega}]_{1qi} \in \Gamma_k$ are incomparable with
respect to $\leq_{1qi}$.

It remains to show that any anti-chain of 1qi-degrees contained in 
$[\alpha]_{mqi}$ must be finite. Consider any anti-chain of 1qi-degrees 
containing the class $[1^i2^j0^{\omega}]_{1qi} \subseteq [\alpha]_{mqi}$.
Every element of this anti-chain that is different from 
$[1^i2^j0^{\omega}]_{1qi}$ is of the form $[1^{i'}2^{j'}0^{\omega}]_{1qi}$,
where either $i < i'$ and $j > j'$, or $i > i'$ and $j < j'$.
Thus, if the anti-chain were infinite, then it would contain at least
$2$ 1qi-degrees, $[\beta]_{1qi}$ and $[\gamma]_{1qi}$,
such that either $\beta$ has the same number of occurrences of $1$
as $\gamma$, or $\beta$ has the same number of occurrences of $2$
as $\gamma$. This is a contradiction as it would imply that
either $\beta \leq_{1qi} \gamma$ or $\gamma \leq_{1qi} \beta$.  
\end{proof}

\subsection{pqi-Reductions}

We now discuss pqi-reductions, which are the most stringent
kind of quasi-isometric reductions considered in the present
work. Pqi-reductions are 1qi-reductions that are surjective;
an example of such a reduction is the mapping
$2m-1 \mapsto 2m$, $2m \mapsto 2m-1$ from 
$(01)^{\omega}$ to $(10)^{\omega}$. We record a few elementary
properties of pqi-reductions.

\begin{lemma}\label{lem:cross.over.lb}
If $f$ is a pqi-reduction and if $x+D = f(x)$
for some $D \geq 1$ and some $x \in \bbbn$, then
there are at least $D$ positions $y > x$ such that 
$f(y) < f(x)$. 
\end{lemma}

\begin{proof}
If $x+D = f(x)$ for some $D \geq 1$, then 
$\{1, \ldots, x+D-1\} \, \backslash \, \{f(1), \ldots, f(x-1)\}$ must contain 
at least $D$ elements as the former set contains $D$ more elements than 
the latter. Thus, for $f$ to be a bijection, there must exist at
least $D$ positions $y > x$ that are mapped by 
$f$ into $\{1, \ldots, x+D-1\} \, \backslash \, \{f(1), \ldots, f(x-1)\}$. 
\end{proof}

We next observe that for any pqi-reduction $f$, 
there is a uniform upper bound on the difference
$x-f(x)$.
 
\begin{proposition}\label{prop:pqi.diff.ub}
If $f$ is a $C$-pqi-reduction, then
for all $x \in \bbbn$, $x - f(x) < 2C^2+1$.
\end{proposition}

\begin{proof}
Assume, by way of contradiction, that there is some
$x \in \bbbn$ such that $x - f(x) \geq 2C^2+1$.
First, suppose that there are at least $C^2+1$ numbers 
$z$ such that $z > x$ and 
$f(z) \in \{f(x)+1,f(x)+2,\ldots,x-1\}$. 
Then, there are at least $C^2+1$ numbers $z'$ such that 
$z' < x$ and $f(z')  > x > f(x)$, among which there is at 
least one $z'_0$ with $f(z'_0) \geq x + C^2+1$.
This would contradict the fact that by the Small Cross-Over
Lemma (Lemma \ref{lem:small.crossover}), 
$z'_0 < x \Rightarrow f(z'_0) \leq f(x)+C^2 < x+C^2$.

Second, suppose that $f$ maps at most $C^2$ numbers 
greater than $x$ into $\{f(x)+1,f(x)+2,\ldots,x-1\}$.
Then, there are at least $C^2+1$ numbers less than
$x$ that are mapped into $\{f(x)+1,f(x)+2,\ldots,x-1\}$
and in particular, there is at least one number
$y < x$ such that $f(y) \geq f(x)+C^2+1$, 
contradicting the Small Cross-over Lemma. 
Thus, for all $x \in \bbbn$, $x-f(x) < 2C^2+1$.   
\end{proof}

Lemma \ref{lem:cross.over.lb} and Proposition 
\ref{prop:pqi.diff.ub} together give a uniform
upper bound on the absolute difference between any
position number and its image under a $C$-pqi-reduction.

\begin{corollary}\label{cor:pqi.absdiff.bound}
If $f$ is a $C$-pqi-reduction, then
for all $x \in \bbbn$, $|x - f(x)| < 2C^2+1$.
\end{corollary}

\begin{proof}
By Condition (b), there cannot be more than $C$ numbers
$y$ such that $y > x$ and $f(y) < f(x)$. Lemma 
\ref{lem:cross.over.lb} thus implies that there cannot
exist any $D > C$ such that $x + D = f(x)$, and so
$f(x) - x \leq C$. Combining the latter inequality with
that in Proposition \ref{prop:pqi.diff.ub} yields
$|x-f(x)| < \max\{C+1,2C^2+1\} = 2C^2+1$.   
\end{proof}

Given any infinite string $\alpha$, it was observed
earlier that by the definitions of pqi, 1qi and 
mqi-reductions, 
$[\alpha]_{pqi} \subseteq [\alpha]_{1qi} \subseteq 
[\alpha]_{mqi}$. In the following example, we give 
instances of strings $\alpha$ where each of the two
subset relations is proper or can be replaced with
the equals relation.

\begin{example}\label{exmp:reduction.separations}
\begin{enumerate}[(a)]
\item\label{eg1} $[\alpha]_{pqi} = [\alpha]_{1qi} = [\alpha]_{mqi}$.
Set $\alpha = 0^{\omega}$.
For any infinite string $\gamma$
such that $\gamma \leq_{mqi} 0^{\omega}$, $\gamma$
can only contain occurrences of $0$, and therefore
$[0^{\omega}]_{pqi} = [0^{\omega}]_{1qi} = [0^{\omega}]_{mqi} = 
\{0^{\omega}\}$.

\item\label{eg2} $[\alpha]_{1qi} = [\alpha]_{mqi}$ and 
$[\alpha]_{pqi} \subset [\alpha]_{1qi}$.
Set $\alpha = (01)^{\omega}$.
First, $(001)^{\omega} \leq_{1qi} (01)^{\omega}$,
as witnessed by the 1qi-reduction $3n-2 \mapsto
4n-3, 3n-1 \mapsto 4n-1, 3n \mapsto 4n$ for 
$n \in \bbbn$. We also have 
$(01)^{\omega} \leq_{1qi} (001)^{\omega}$
via the 1qi-reduction $2n-1 \mapsto 3n-2, 
2n \mapsto 3n$ for $n \in \bbbn$.
However, $(001)^{\omega} \notin [(01)^{\omega}]_{pqi}$
because the density of $0$'s and $1$'s in
the two strings are different, making it impossible
to construct a permutation reduction between
them. More formally, if there were a pqi-reduction
from $(001)^{\omega}$ to $(01)^{\omega}$, then by
Corollary \ref{cor:pqi.absdiff.bound}, there would be a constant
$D$ such that for each $n$, the first $3n$ positions of $(001)^{\omega}$ 
are mapped into the first $3n+D$ positions of $(01)^{\omega}$.
But the first $3n$ positions of $(001)^{\omega}$ contain $2n$ 
occurrences of $0$ while the first $3n+D$ positions of
$(01)^{\omega}$ contain at most $\left\lceil 1.5n+\frac{D}{2} \right\rceil$ 
occurrences of $0$, and for large enough $n$, one has
$2n > \left\lceil 1.5n+\frac{D}{2} \right\rceil$. Hence, no pqi-reduction from
$(001)^{\omega}$ to $(01)^{\omega}$ can exist, and
so $[\alpha]_{pqi} \subset [\alpha]_{1qi}$.

To see that 
$[(01)^{\omega}]_{mqi} \subseteq [(01)^{\omega}]_{1qi}$,
we first note that any string that is mqi-reducible to 
$(01)^{\omega}$ (or to any other recursive string) must be 
recursive. Thus, if $\beta \leq_{mqi} (01)^{\omega}$, then
a 1qi-reduction from $\beta$ to $(01)^{\omega}$ can be
constructed by mapping the $n$-th position of $\beta$
to the position of the matching letter in the $n$-th 
occurrence of $01$ in $(01)^{\omega}$. Next, suppose 
that $f$ is a $C$-mqi-reduction from $(01)^{\omega}$
to $\beta$. By Corollary \ref{cor:nary.seq}, $f$ maps
the positions of $(01)^{\omega}$ to a sequence
of positions of $\beta$ that contains $0$ and $1$  
every $2C$ positions. Thus, a 1qi-reduction can be
constructed from $(01)^{\omega}$ to $\beta$ by mapping, 
for each $n$, the $(2n-1)$-st and $(2n)$-th positions 
of $(01)^{\omega}$ to the positions of the first occurrence
of $0$ and first occurrence of $1$ respectively in the
interval $[2C(n-1)+1, 2Cn]$ of positions of $\beta$.
Therefore, $\beta \in [(01)^{\omega}]_{1qi}$.

\item\label{eg3} $[\alpha]_{1qi} \subset [\alpha]_{mqi}$ and 
$[\alpha]_{pqi} = [\alpha]_{1qi}$.
Set $\alpha = 10^{\omega}$.
We recall from the proof of Proposition \ref{prop:union.asc.chain}
that $[10^{\omega}]_{pqi}$ and $[10^{\omega}]_{1qi}$
consist of all binary strings with a single occurrence
of $1$, while $[10^{\omega}]_{mqi}$ consists of all 
binary strings with a finite, positive number of occurrences
of $1$. Thus, $[10^{\omega}]_{pqi} = [10^{\omega}]_{1qi}$
and $[10^{\omega}]_{pqi} \neq [10^{\omega}]_{mqi}$.

\item\label{eg4} 
$[\alpha]_{pqi} \subset [\alpha]_{1qi} \subset [\alpha]_{mqi}$.
Set $\alpha = (0^n1)_{n=1}^{\infty}$, the concatenation
of all strings $0^n1$ where $n \in \bbbn$.
Then, $\beta = (0^n11)_{n=1}^{\infty} \in [\alpha]_{mqi}$;
however, $\beta \notin [\alpha]_{1qi}$ as each pair of 
adjacent positions of $1$'s in $\beta$ must be 
mapped to distinct positions of $1$'s in 
$\alpha$, but the distance between the $n$-th and $(n+1)$-st 
occurrences of $1$ in $\alpha$ increases linearly with $n$, 
meaning that Condition (a) cannot be satisfied.

To construct an mqi-reduction from $\beta$ to $\alpha$,
map the positions of the substring $0^n11$ of $\beta$
to the positions of the substring $0^n1$ of $\alpha$
as follows: for $k \in \{1,\ldots,n\}$, the position of the 
$k$-th occurrence of $0$ in $0^n11$ is mapped to that of the $k$-th
occurrence of $0$ in $0^n1$, while the two positions of $1$'s
in $0^n11$ are mapped to the position of the single $1$ in $0^n1$.
For an mqi-reduction from $\alpha$ to $\beta$,
for each substring $0^n1$ of $\alpha$ and each substring $0^n11$
of $\beta$, the positions of $0^n$ in $0^n1$ are mapped to the 
corresponding positions of $0^n$ in $0^n11$, while the 
position of $1$ in $0^n1$ is mapped to the position of the first
occurrence of $1$ in $0^n11$. Thus, $\beta \in [\alpha]_{mqi}$.

Furthermore, $\gamma = 1(0^n1)_{n=1}^{\infty} \in [\alpha]_{1qi}$
but $\gamma \notin [\alpha]_{pqi}$. The reason for $\gamma$ not
being pqi-reducible to $\alpha$ is similar to that given in 
Example (\ref{eg2}). If such a pqi-reduction did exist, then by 
Corollary \ref{cor:pqi.absdiff.bound}, 
there would exist a constant $D$ such that
for all $n$, the first $1+\sum_{k=1}^n (k+1) = 1+\frac{n(n+3)}{2}$
positions of $\gamma$ are mapped into the first 
$1+\frac{n(n+3)}{2}+D$ positions of $\alpha$. But the first
$1+\frac{n(n+3)}{2}$ positions of $\gamma$ contain $n+1$ occurrences
of $1$ and for large enough $n$, the first $1+\frac{n(n+3)}{2}+D$ 
positions of $\alpha$ contain at most $n$ occurrences of $1$.
Hence, no pqi-reduction from $\gamma$ to $\alpha$ is possible. 

For a 1qi-reduction from $\gamma$ to $\alpha$,
map the starting position of $\gamma$, where the letter $1$
occurs, to the first occurrence of $1$ in $\alpha$. 
For subsequent positions of $\gamma$, for each $n \geq 1$, 
the set of positions of $\gamma$ where the substring
$0^n1$ occurs can be mapped in a one-to-one
fashion into the set of positions of $\alpha$ where the
substring $0^{n+1}1$ occurs. To see that $\alpha$ is 
1qi-reducible to $\gamma$, it suffices to observe that
$\alpha$ is a suffix of $\gamma$, so one can map the positions
of $\alpha$ in a one-to-one fashion to the positions
of the suffix of $\gamma$ corresponding to $\alpha$.
\end{enumerate}
\end{example} 

Proposition \ref{prop:equal.mqi.pqi.1qi.degrees} extends the first 
example in Example \ref{exmp:reduction.separations}
by characterising all recursive strings whose pqi, 1qi and 
mqi-degrees all coincide. In fact, there are only 
$|\Sigma|$ many such strings: those of the form 
$a_i^{\omega}$, where $a_i \in \Sigma$. We call the
pqi, 1qi and mqi-degrees of such strings {\em trivial}.

\begin{definition}\label{defn:trivial.degree}
The pqi, 1qi and mqi-degrees of each string
$a_i^{\omega}$, where $a_i \in \Sigma$, will be
called {\em trivial} pqi, 1qi and mqi-degrees 
respectively. 
\end{definition}   

\begin{restatable}{proposition}{trivialdegree}
\label{prop:equal.mqi.pqi.1qi.degrees}
If, for some recursive string $\alpha$, 
$[\alpha]_{pqi} = [\alpha]_{1qi} = [\alpha]_{mqi}$,
then all three degree classes are trivial.
\end{restatable}

\begin{proof}
Suppose $\Sigma = \{a_1,\ldots,a_{l}\}$, where without loss
of generality it may be assumed that $l \geq 2$. 
Fix any string $\alpha$ whose pqi, 1qi and mqi-degrees all 
coincide, and assume, by way of contradiction, that at least two
distinct letters occur in $\alpha$.
Now we consider the following case distinction.
\begin{description}
\setlength{\labelwidth}{-1.5cm}
\setlength{\itemindent}{-1.95cm}
\item[Case (i):] \textit{There 
is some $a_i$ that occurs a finite,
positive number of times in $\alpha$.}
Let $\beta = a_i\alpha$. Then, $\beta \notin [\alpha]_{pqi}$
as $\beta$ contains exactly one more occurrence of $a_i$ than
$\alpha$, so a colour-preserving bijection between $\alpha$
and $\beta$ cannot exist. On the other hand, $\beta \in
[\alpha]_{mqi}$.
First, $\beta \leq_{mqi} \alpha$ via a reduction that maps the
first position of $\beta$ to the position of the first occurrence
of $a_i$ in $\alpha$; the positions of the suffix $\alpha$ of 
$\beta$ are then mapped in a one-to-one fashion to the corresponding
positions of $\alpha$. Second, $\alpha \leq_{1qi} \beta$
as $\alpha$ is a suffix of $\beta$. 
  
\item[Case (ii):] \textit{Each letter in $\alpha$ occurs infinitely
often.}
\begin{description}
\setlength{\labelwidth}{-1.5cm}
\setlength{\itemindent}{-2.35cm}
\item[Subcase (a):] \textit{There is some letter $a_i$ occurring
in $\alpha$ such that there is no uniform upper bound on the
distance between successive occurrences of $a_i$.}
Consider the string $\beta = a_i\alpha$.
The same argument as in Case (i) shows that
$\beta \in [\alpha]_{mqi}$. We show that 
$\beta \not\leq_{pqi} \alpha$. If there were a pqi-reduction from 
$\beta$ to $\alpha$, then by Corollary \ref{cor:pqi.absdiff.bound}, 
there would exist a constant $D$ such that for all 
$n$, the first $n$ positions of $\beta$ are mapped into the
first $n+D$ positions of $\alpha$. Now pick $k$ large 
enough so that the distance between the $k$-th and $(k+1)$-st 
occurrences of $a_i$ in $\alpha$ is greater than $D+1$.
Let $n$ be the position number of the $(k+1)$-st occurrence of
$a_i$ in $\beta$; then $n-1$ is the position number of the
$k$-th occurrence of $a_i$ in $\alpha$. By Corollary
\ref{cor:pqi.absdiff.bound}, the positions of the first 
$k+1$ occurrences of $a_i$ in $\beta$ must be mapped into
the first $n+D$ positions of $\alpha$, but since the 
$(k+1)$-st occurrence of $a_i$ in $\alpha$ is at a position
greater than $n-1+D+1 = n+D$, it follows that such a 
mapping cannot be one-to-one. 
Hence, $\beta \notin [\alpha]_{pqi}$.
      
\item[Subcase (b):] \textit{There is some $D > 0$ such that
for every letter $a_i$ occurring in $\alpha$, the distance
between successive occurrences of $a_i$ is at most $D$.}
Let $D$ be the largest distance between successive occurrences
of the same letter in $\alpha$.
As $\alpha$ is recursive, one can construct a 1qi-reduction
from $\alpha$ to $(a_1\cdots a_l)^{\omega}$ by mapping the 
$n$-th position of $\alpha$ to the position in the $n$-th
occurrence of $a_1\cdots a_l$ where the corresponding letters
match. Furthermore, a 1qi-reduction from 
$(a_1\cdots a_l)^{\omega}$ to $\alpha$ may be constructed
as follows. Let $m$ be the first position of $\alpha$ such that
every letter in $\alpha$ occurs at least once before position $m+1$. 
By the choice of $D$, any interval of $D$ positions of $\alpha$
starting after position $m$ contains every letter. We define
a mapping from $(a_1\cdots a_l)^{\omega}$ to $\alpha$
such that the position of the first occurrence of each letter in
$(a_1\cdots a_l)^{\omega}$ is mapped to the position of the
first occurrence of the corresponding letter in $\alpha$, 
and for $k \geq 2$, for the $k$-th occurrence of $a_1\cdots a_l$,
the position of each letter is mapped to the position of the first 
occurrence of the corresponding letter in the interval 
$[m+(k-1)D+1, m+kD]$ of positions of $\alpha$. This mapping is
one-to-one and also satisfies Conditions (a) and (b) for the
constant $\max\{m+D,2D\}$, so it is indeed a 1qi-reduction.
Therefore, $[(a_1\cdots a_l)^{\omega}]_{1qi} = [\alpha]_{1qi}
= [\alpha]_{mqi}$. 

Now consider $\beta = (a_1a_1a_2\cdots a_l)^{\omega}$. A 1qi-reduction
from $\beta$ to $(a_1\cdots a_l)^{\omega}$ can be constructed
by mapping, for the $k$-th occurrence of $a_1a_1a_2\cdots a_l$,
the position of the first occurrence of $a_1$ to the position of 
$a_1$ in the $(2k-1)$-st occurrence of $a_1\cdots a_l$,
and the positions of subsequent letters to the positions of
matching letters in the $(2k)$-th occurrence of $a_1\cdots a_l$.
A 1qi-reduction from $(a_1\cdots a_l)^{\omega}$ to $\beta$
can also be constructed by mapping the positions of the
$k$-th occurrence of $(a_1\cdots a_l)$ to positions of matching
letters in the $k$-th occurrence of $(a_1a_1a_2\cdots a_l)$. 
Thus, $\beta \in [\alpha]_{1qi}$. On the other hand, a proof 
entirely similar to that in part (\ref{eg2}) of Example 
\ref{exmp:reduction.separations} shows that
$\beta \notin [\alpha]_{pqi}$. We conclude that 
$[\alpha]_{pqi} \neq [\alpha]_{1qi}$. 
\qedhere  
\end{description}
\end{description}
\end{proof}

We observe next that every non-trivial pqi degree must be infinite.

\begin{proposition}\label{prop:non.trivial.degrees.infinite}
All non-trivial pqi-degrees are infinite.
\end{proposition} 

\begin{proof}
Suppose that at least two distinct letters occur in $\alpha$.
Fix a letter, say $a_1$, that occurs infinitely often in 
$\alpha$. Let $a_2$ be a letter different from $a_1$ that
occurs in $\alpha$. For each $n \in \bbbn$, let 
$\beta_n = a_1^n a_2 \alpha^{(n+1)}$, where
$\alpha^{(n+1)}$ is obtained from $\alpha$ by removing
the first occurrence of $a_2$ as well as the first $n$
occurrences of $a_1$. Since $\beta_n$ is built from
$\alpha$ by permuting the letters occurring at a finite set of 
positions of $\alpha$, $\beta_n \in [\alpha]_{pqi}$. As the 
$\beta_n$'s are all distinct, it follows that $[\alpha]_{pqi}$ 
is indeed infinite.   
\end{proof}

We close this subsection by illustrating an application
of Proposition \ref{prop:non.trivial.degrees.infinite},
showing that if the mqi-degree of $\alpha$ contains at 
least two distinct strings such that one is 1qi-reducible 
to the other, then the first string is 1qi-reducible to 
infinitely many strings in $[\alpha]_{mqi}$. 

\begin{proposition}\label{prop:mqi.1qi.infinite}
If there exist distinct $\beta \in [\alpha]_{mqi}$ and
$\gamma \in [\alpha]_{mqi}$ such that $\beta \leq_{1qi} \gamma$,
then $\beta$ is 1qi-reducible to infinitely many strings 
in $[\alpha]_{mqi}$.
\end{proposition}

\begin{proof}
Suppose that $\beta \leq_{1qi} \gamma$ and $\beta \neq \gamma$
for some $\beta \in [\alpha]_{mqi}$ and 
$\gamma \in [\alpha]_{mqi}$. Then, $[\alpha]_{mqi}$ is non-trivial,
so by Proposition \ref{prop:non.trivial.degrees.infinite}, 
$[\gamma]_{pqi}$ is infinite. Since $[\gamma]_{pqi} \subseteq
[\gamma]_{1qi}$, $[\gamma]_{1qi}$ is also infinite. Thus,
$\beta$ is 1qi-reducible to each of the infinitely many strings in 
$[\gamma]_{1qi}$.  
\end{proof}

\subsection{The Partial Order of All mqi-Degrees}

As discussed earlier, Khoussainov and Takisaka \cite{Khoussainov17}
observed that for any alphabet $\Sigma = \{a_1,\ldots,a_l\}$,
the partial order $\Sigma^{\omega}_{qi}$ has a greatest element
equal to $[(a_1\cdots a_l)^{\omega}]_{qi}$. Their proof also extends
to the partial order of all recursive mqi-degrees, showing that 
for each recursive string $\alpha$, $[\alpha]_{mqi} \leq_{mqi} 
[(a_1\cdots a_l)^{\omega}]_{mqi}$. 
We next prove that there is a pair of recursive mqi-degrees whose 
join is precisely the maximum recursive mqi-degree 
$[(a_1\cdots a_l)^{\omega}]_{mqi}$.

\begin{proposition}\label{prop:mqi.join}
Suppose that $\Sigma = \{a_1,\ldots,a_l\}$. Then, there exist
two distinct infinite strings $\alpha$ and $\beta$ such that
$[(a_1\cdots a_l)^{\omega}]_{mqi}$ is the unique recursive
common upper bound of $[\alpha]_{mqi}$ and of $[\beta]_{mqi}$ 
under $\leq_{mqi}$.
\end{proposition}

\begin{proof}
Let $\alpha = (a_1)^{\omega}$ and $\beta = (a_2a_3\cdots a_l)^{\omega}$.
Suppose that for some recursive string $\gamma$, $\alpha \leq_{mqi} \gamma$
via a $C$-mqi-reduction. Since $a_1$ is the only letter occurring in 
$\alpha$, Condition (a) implies that there must be at least one 
occurrence of $a_1$ in $\gamma$ every $C$ positions. Similarly, if 
$\beta \leq_{mqi} \gamma$ via a $C'$-mqi-reduction, then for each 
$a_i$ with $i \geq 2$, since $a_i$ occurs every $l-1$ positions, it must 
also occur in $\gamma$ every $C'\cdot(l-1)$ positions. 
Hence, there exists a constant $C''$ such that every substring of 
$\gamma$ of length $C''$ contains at least one occurrence of $a_i$ for 
every $i \in \{1,\ldots,l\}$, and therefore $(a_1\cdots a_l)^{\omega}
\leq_{mqi} \gamma$. Since $\gamma \leq_{mqi} (a_1\cdots a_l)^{\omega}$
follows from the proof of \cite[Proposition II.1]{Khoussainov17},
one has $\gamma \in [(a_1\cdots a_l)^{\omega}]_{mqi}$, as required.      
\end{proof}

Khoussainov and Takisaka \cite{Khoussainov17} showed that the partial order 
$\Sigma^{\omega}_{qi}$ is not dense. In particular, given any 
distinct $a_i,a_j \in \Sigma$, there is no 
element $[\beta]_{qi}$ that is strictly between the minimal element
$[(a_j)^{\omega}]_{qi}$ and the ``atom'' $[a_i(a_j)^{\omega}]_{qi}$
\cite[Proposition II.1]{Khoussainov17}. 
The next theorem shows similarly that the partial order
$\Sigma^{\omega}_{mqi}$ is non-dense with respect to pairs
of mqi-degrees.
Furthermore, it implies that $\Sigma^{\omega}_{mqi}$ is neither a join-semilattice nor a meet-semilattice.

\begin{restatable}{theorem}{mqinondense}\label{thm:mqi.non.dense}
There exist two pairs $(\alpha,\beta)$ and $(\gamma,\delta)$ of 
recursive strings such that both $\alpha$ and $\beta$ are
mqi-reducible to $\gamma$ as well as mqi-reducible to $\delta$,
but there is no string $\xi$ such that
$\alpha \leq_{mqi} \xi, \beta \leq_{mqi} \xi, \xi \leq_{mqi} \gamma$
and $\xi \leq_{mqi} \delta$. 
\end{restatable}

\begin{proof}
Let $\Sigma = \{0,1\}$. Define the strings
\begin{align*}
\alpha &= \sigma_1\sigma_2\ldots,~~\text{where $\sigma_i = (01)^{2^{2^i}}0^i1^i$}\,; \\
\beta &= \tau_1\tau_2\ldots,~~\text{where $\tau_i = (01)^{2^{2^i}}1^i0^i$}\,; \\
\gamma &= \mu_1\mu_2\ldots,~~\text{where $\mu_i = (01)^{2^{2^i}}0^i$}\,; \\
\delta &= \nu_1\nu_2\ldots,~~\text{where $\nu_i = (01)^{2^{2^i}}1^i$}\,. 
\end{align*}
We first show that $\alpha \leq_{mqi} \gamma$.
For each $i \in \bbbn$, define the following intervals of positions.
\begin{align*}
K_i &= [k_i, k_i+2^{2^i+1}-1] \text{ is the interval of positions of
the substring $(01)^{2^{2^i}}$ of $\sigma_i$ in $\alpha$.} \\
R_i &= [r_i, r_i+2i-1] \text{ is the interval of positions of
the substring $0^i1^i$ of $\sigma_i$ in $\alpha$.} \\
L_i &= [l_i, l_i+2^{2^i+1}-1] \text{ is the interval of positions of
the substring $(01)^{2^{2^i}}$ of $\mu_i$ in $\gamma$.} \\
L'_i &= [l_i+2^{2^i+1}, l_i+2^{2^i+1}+i-1] \text{ is the interval of
positions of the substring $0^i$ of $\mu_i$ in $\gamma$.}
\end{align*}
Define an mqi-reduction $g$ from $\alpha$ to $\gamma$ as follows.
For $i \in \bbbn$,
\begin{align*}
g(k_i+4w+2u+x) &= l_i+2(i-1)+2w+x\,,~~0 \leq w \leq i-2, u, x \in \{0,1\}\,; \\
g(k_i+m) &= l_i+m\,,~~4i-4\leq m \leq 2^{2^i+1}-1\,; \\
g(r_i+m) &= l_i+2^{2^i+1}+m\,,~~0 \leq m \leq i-1\,; \\ 
g(r_i+i+m) &= l_{i+1}+2m+1\,,~~0 \leq m \leq i-1. 
\end{align*}
In other words, $g$ maps
the positions of the substring $(01)^{2^{2^i}}$ of $\sigma_i$ in $\alpha$ to
the last $2^{2^i+1}-2(i-1)$ positions of
the substring $(01)^{2^{2^i}}$ of $\mu_i$ in $\gamma$.
The mapping is as follows.
Each of the $i$-th to the $2(i-1)$-th pairs $01$ of $\mu_i$ is
an image of two consecutive pairs $01$ of $\sigma_i$.
Then, the last $2^{2^i+1}-4(i-1)$ positions of
the substring $(01)^{2^{2^i}}$ of $\mu_i$ in $\gamma$
are mapped from the last $2^{2^i+1}-4(i-1)$ positions of
the substring $(01)^{2^{2^i}}$ of $\sigma_i$ in $\alpha$
in a one-to-one fashion.
Furthermore, $g$ maps the positions of the substring $0^i$ of $\sigma_i$ in $\alpha$
to the positions of the substring $0^i$ of $\mu_i$ in $\gamma$ in a one-to-one fashion.
Then, the positions of the substring $1^i$ of $\sigma_i$ in $\alpha$ is mapped to
the first $i$ positions of $1$ in $\mu_{i+1}$ in $\gamma$,
ending at the $2i$-th position of
the substring $(01)^{2^{2^{i+1}}}$ of $\mu_{i+1}$ in $\gamma$.
Thus, $g$ is a $4$-mqi-reduction from $\alpha$ to $\gamma$.

A similar mqi-reduction can be
constructed from $\beta$ to $\gamma$. In this case,
the mqi-reduction maps the interval $K_i$ to the prefix of
the interval $L_i$ where the last $2i$ positions are cut off
and each of the last $i$ pairs of positions of the prefix is the image of
two consecutive pairs of positions of $K_i$.
By symmetrical constructions, one obtains mqi-reductions
from $\alpha$ to $\delta$ as well as
from $\beta$ to $\delta$.

Now assume, for the sake of contradiction, that there is a string $\xi$
and there are mqi-reductions 
$f_1$ from $\alpha$ to $\xi$, $f_2$ from $\beta$ to $\xi$,
$f_3$ from $\xi$ to $\gamma$ and $f_4$ from $\xi$ to $\delta$
with constants $C_1,C_2,C_3$ and $C_4$
respectively. Set $C = \max\{C_1,C_2,C_3,C_4\}$ and fix
some $n > 2C^7+1$. 

For $i \in \bbbn$,
let $K'_i = [k_i + C^2+2, k_i + 2^{2^{i}+1}-3-C^2]$ be the interval 
obtained from $K_i$ by removing the first and last $C^2+2$ positions.
We make the following observation.

\begin{claim}\label{clm:subinterval.map}
For all positions $m \in K'_n$, 
for $i \in \{1,2\}$ and 
$j \in \{3,4\}$, $f_j(f_i(m)) \in L_n$.
\end{claim}

\begin{claimproof}
To simplify the subsequent argument,
we assume, without loss of generality, that at least one
position of $\gamma$ lies in the intersection of 
$\bigcup_{i<n} L_i \cup L'_{i}$ and the range of $f_3\circ f_1$.  
We first note that the map $f_3 \circ f_1$ is a $C^2$-mqi-reduction
from $\alpha$ to $\gamma$.
Since the length of the interval $[1,l_n-1]$ is 
$\sum_{i=1}^{n-1} (i + 2^{2^{i}+1}) \leq (n-1)\cdot(n-1 + 2^{2^{n-1}+1})
\leq 3n\cdot2^{2^{n-1}}$, there are at most $3C^2n\cdot2^{2^{n-1}}$
positions of $K_n$ in $\alpha$ that are mapped into the interval 
$[1,l_n-1]$ of $\gamma$.
Since $|K_n| = 2^{2^n+1} > 3C^2n\cdot2^{2^{n-1}}$ and there
is, by assumption, at least one point in the range of $f_3 \circ f_1$
that lies in $\bigcup_{i<n} L_i \cup L'_{i}$, 
it follows from the fact that there cannot be gaps larger than $C^2$ 
in the range of $f_3\circ f_1$ that at least one position of $K_n$, say 
$m_0$, must be mapped under $f_3 \circ f_1$ into the interval $L_n$ of 
$\gamma$. 

Thus, $f_3\circ f_1$ cannot map any position of $K_n$ into $L_{n-1}$. 
For, if there were a least such position $m_1 \in K_n$ and $m_1 < m_0$, then by Condition (a) and using the fact that $n-1 > 2C^2$, $f_3\circ f_1$ must map 
$m_1+1$ into either $L_{n-1}$ or to one of the first $C^2$ positions of 
$L'_{n-1}$. In the former case, $f_3\circ f_1$ must map
$m_1+2$ into either $L_{n-1}$ or to one of the first $C^2$ positions
of $L'_{n-1}$ and the same argument can be iterated. In the 
latter case, the letters of $\alpha$ at positions $m_1+1$ and $m_1+2$ must be 
$0$ and $1$ respectively, which implies that $f_3\circ f_1$ must
map position $m_1+2$ into $L_{n-1}$, and the same argument can
again be iterated. Iterating the argument, it would then follow that
$f_3\circ f_1$ maps $m_0$ into $L_{n-1}$ or to one of the first
$C^2$ positions of $L'_{n-1}$, a contradiction. A similar argument
applies in the case that $m_1 > m_0$. 

Furthermore, if $f_3\circ f_1$ maps some position $m$ of $K_n$ into 
$L'_{n-1}$, then $m$ must be contained in the first $C^2+2$ positions 
of $K_n$. For, suppose that $m$ occurs after the first $C^2+2$ positions
of $K_n$, then each of the first two positions of $K_n$ is at least 
$C^2+1$ positions away from $m$. So, by Condition (b), $f_3\circ f_1$
must map the first two positions of $K_n$ to some position of 
$\gamma$ before $(f_3\circ f_1)(m)$, which lies in $L'_{n-1}$; but
this is impossible since the letter in the second position of $K_n$ is $1$ 
and $L'_{n-1}$ contains only $0$'s and, as was shown 
earlier, no position of $K_n$ is mapped into $L_{n-1}$. Therefore,
$f_3\circ f_1$ maps at most $C^2+2$ positions of 
$K_n$ in $\alpha$ into the interval $L'_{n-1}$ of $\gamma$, and these
positions must occur within the first $C^2+2$ positions of $K_n$. 
Similarly, $f_3\circ f_1$ maps at most $C^2+2$ positions of $K_n$
in $\alpha$ into the interval $L'_n$ of $\gamma$, and these positions
must occur within the last $C^2+2$ positions of $K_n$. 
Summing up, for each position $m$ in the interval 
$K'_n = [k_n + C^2+2, k_n + 2^{2^{n}+1}-3-C^2]$, $f_3(f_1(m)) \in L_n$.
Similar arguments show that $f_j(f_i(m)) \in L_n$ for 
$(i,j) \in \{(1,4),(2,3),(2,4)\}$.
\end{claimproof}

Now define the sets $H_{i} = f_1(K'_i) \cup f_2(K'_i)$ for $i \in \bbbn$. 
We show that the sets $H_n$ and $H_{n+1}$ are non-overlapping by
proving $\max(H_n) < \min(H_{n+1})$.
By Claim \ref{clm:subinterval.map} above, for all $m \in H_n$ and 
$j \in \{3, 4\}$ we have $f_j(m) \in L_n$.
Then, we have $f(\min(H_{n + 1})) - f(\max(H_n)) \geq 
\min(L_{n + 1}) - \max(L_n) = n + 1 > 2C^7 + 2 > C^2$.
So by the \hyperref[lem:small.crossover]{Small Cross-Over Lemma},
$\min(H_{n + 1}) > \max(H_n)$.

Consider the interval $[\min(H_n),\max(H_n)]$ in the domain of 
$\xi$. By Claim \ref{clm:subinterval.map}, $f_3$ (resp.~$f_4$) maps each element of 
$H_n$ into $L_n$. Fix any other position $z$ in the interval.
Then, $f_3$ cannot map $z$ into $L'_n$, which is the set of positions
in $\gamma$ of the string $0^n$. To see this, we note that if
$\ell$ and $\ell+1$ are the two largest values of $K_n$, then 
$\ell$ is at least $C^2+1$ more than the value $x$ such that 
$f_i(x) = \max(H_n)$ for some $i \in \{1,2\}$, and so by 
Condition (b), $z + C < \max(H_n)+C < f_{k}(\ell+1)$
for $k \in \{1,2\}$. Thus, $f_3(z) < f_3(f_{k}(\ell+1))$ for
$k \in \{1,2\}$.
Furthermore, by applying Condition (a) repeatedly to $f_3$ and then to $f_k$,
we have $d(f_3(f_k(\ell + 1)), f_3(f_k(\max(K'_n)))) 
\leq C \cdot d(f_k(\ell + 1), f_k(\max(K'_n))) \leq C^2 \cdot (C^2 + 2) = C^4 + 2C^2$.
Since $f_3(f_k(\max(K'_n))) \in L_n$ and we fixed $n > 2C^7+1$,
then $f_3(f_{k}(\ell+1)) 
\notin L_{n+1}$. Furthermore, the letter at position 
$f_3(f_k(\ell+1))$ of $\gamma$ is $1$. Thus, $f_3(z)$
cannot lie in $L'_n$ as there is no occurrence of $1$ 
in $L'_n$. A similar argument, using position $\ell$ rather than
position $\ell+1$, shows that $f_4(z)$ cannot lie in $L'_n$.
One can also prove similarly that none of the positions in
the interval $[\min(H_{n+1}),\max(H_{n+1})]$ is mapped by
$f_3$ or $f_4$ into the interval $L'_n$. 

Next, we consider the positions of $\xi$ between $\max(H_n)$
and $\min(H_{n+1})$.
Since none of the positions of $\xi$ in the union 
$[\min(H_{n}),\max(H_n)] \cup [\min(H_{n+1}),\max(H_{n+1})]$
is mapped by $f_3$ into $L'_{n}$ and $L'_{n}$ is an
interval of length $n > 2C^7$, Lemma \ref{lem:image.point.density}
implies that there are at least $\lfloor\frac{n}{C_3}\rfloor$
positions of $\xi$ between $H_n$ and $H_{n+1}$ which are
mapped into $L'_n$.

\begin{claim}\label{clm:substring.zeroes.c3}
The string $\xi$ contains a substring of $0$'s (resp. $1$'s) of length
$\Omega(C^4)$ between $H_n$ and $H_{n+1}$ such that all
positions of this substring are mapped by $f_3$ (resp. $f_4$) into $L'_n$.
\end{claim}

\begin{claimproof}
Let $m_1,\ldots,m_{\ell}$ be all the positions of $L'_{n}$ in the  
range of $f_3$, where $m_1 < m_2 < \ldots < m_{\ell}$.
Since $L'_{n}$ is an interval of length $n$,
Lemma \ref{lem:image.point.density} implies that
$\ell \geq \lfloor\frac{n}{C_3}\rfloor$.
Let $P = f_3^{-1}(L'_n) \sm f_3^{-1}(\{m_i: 1 \leq i \leq C_3\}
\cup \{m_i: \ell-C_3+1 \leq i \leq \ell\})$ be the set of positions
of $\xi$ which are mapped into $L'_n$ but not to any of the
first $C_3$ or the last $C_3$ positions of $L'_n \cap \text{range}(f_3)$.
By Lemma \ref{lem:collision}, 
\[
\begin{aligned}
|f_3^{-1}(\{m_i: 1 \leq i \leq C_3\}
\cup \{m_i: \ell-C_3+1 \leq i \leq \ell\})| \leq
2C^2\,,
\end{aligned}
\] 
and thus 
\[
\begin{aligned}
|P| &\geq \ell-2C^2 \geq \left\lfloor\frac{n}{C_3}\right\rfloor - 2C^2 
= \Omega(C^6),
\end{aligned}
\]   
where we have used the fact that $n > 2C^7$.
The set $P$ is split into at most $2C^2$ groups, each included in an 
interval not containing any position in 
$f_3^{-1}(\{m_i: 1 \leq i \leq C_3\}
\cup \{m_i: \ell-C_3+1 \leq i \leq \ell\})$, and so by the pigeonhole 
principle there is an interval between $H_n$ and $H_{n+1}$
containing at least $\frac{\Omega(C^6)}{2C^2} = \Omega(C^4)$ positions
that are mapped by $f_3$ into $L'_n$, but no position in this
interval is mapped to any of the first $C_3$ or the last $C_3$ positions 
of $L'_n$. Let $m'$ (resp.~$m''$) be the minimum (resp.~maximum) position 
in this interval which is mapped to a position in $L'_n$. If there
were a least position $m''' \in [m',m'']$ such that $f_3(m''') \notin L'_n$, 
then by the choice of $m'$ and $m''$, $f_3(m''') \in L_n \cup L_{n+1}$,
but this is impossible as it would imply that 
$d(f_3(m'''),f_3(m'''-1)) > C_3$, contradicting Condition (a). 
Thus, $[m',m'']$ is an interval between $H_n$ and $H_{n+1}$
of length $\Omega(C^4)$ such that $f_3([m',m'']) \subseteq L'_n$.
An analogous argument, replacing $f_3$ by $f_4$ (thereby considering
the mapping from $\xi$ to $\delta$), shows that $\xi$
contains a substring of $1$'s of length $\Omega(C^4)$
between $H_n$ and $H_{n+1}$ such that $f_4$ maps all positions
of this substring into $L'_n$.
\end{claimproof}

It is shown next that between $H_n$ and $H_{n+1}$, there cannot
exist two $\Omega(C^4)$-long substrings of $0$'s (resp.~$1$'s) 
such that an $\Omega(C^4)$-long substring of $1$'s (resp.~\linebreak[2]$0$'s) 
lies between them.

\begin{claim}\label{clm:no.interm.long.string}
There cannot exist between $H_n$ and $H_{n+1}$ two substrings 
$\delta_1 \in \{0\}^*$ and $\delta_2 \in \{0\}^*$ of $\xi$ with 
$|\delta_1| = \Omega(C^4)$ and $|\delta_2| = \Omega(C^4)$ such that
some $\delta_3 \in \{1\}^*$ with $|\delta_3| = \Omega(C^4)$
is a substring of $\xi$ between $\delta_1$ and $\delta_2$.
The same statement holds when $\{1\}^*$ is interchanged with 
$\{0\}^*$.
\end{claim}

\begin{claimproof}
Recall that for $i \in \bbbn$, 
$[r_i, r_i+2i-1]$ is the interval of positions of $\alpha$ 
(resp.~$\beta$) occupied by the substring $0^i1^i$ 
(resp.~$1^i0^i$) of $\sigma_i$ (resp.~$\tau_i$); denote
this interval by $R_i$.
Let $I_n$ denote the interval of positions of $\xi$ between (exclusive)
$H_n$ and $H_{n+1}$. We give a proof for the case where 
$\delta_1 \in \{0\}^*$, $\delta_2 \in \{0\}^*$ and 
$\delta_3 \in \{1\}^*$.
Since $\text{range}(f_1)$ cannot contain
gaps of size more than $C_1$, there are $\frac{\Omega(C^4)}{C}
= \Omega(C^3)$ positions of $\delta_1$ (resp.~$\delta_2,\delta_3$) 
that belong to $\text{range}(f_1)$.
We observe two facts: first, no position of $\alpha$ before $K_n$ or 
after $K_{n+1}$ is mapped by $f_1$ into $I_n$; second, $f_1$ maps at 
most $\bigO(C^2)$ positions in 
$(K_n \sm K'_n) \cup (K_{n+1} \sm K'_{n+1})$ into $I_n$. These two 
facts imply that $f_1$ maps $\Omega(C^3)$ positions of $R_n$ into the 
interval occupied by $\delta_1$ (resp.~$\delta_2,\delta_3$). 
But then $f_1$ would have to map two consecutive positions of
$R_n$ occupied by $0$'s to positions in $\xi$ that are at least
$|\delta_3| = \Omega(C^4)$ positions apart, contradicting Condition
(a). Hence, no such substrings $\delta_1,\delta_2$ and $\delta_3$ 
can exist.
\end{claimproof}

Based on Claims \ref{clm:substring.zeroes.c3} and \ref{clm:no.interm.long.string}, there
are exactly two maximal intervals $J_1$ and $J_2$, each
of length $\Omega(C^4)$, such that the substrings of $\xi$ 
occupied by $J_1$ and $J_2$ belong to $\{0\}^*$ and 
$\{1\}^*$ respectively. Then, $f_1$ maps $\Omega(C^3)$
positions of $[r_n, r_n+n-1]$ into $J_1$ and 
$\Omega(C^3)$ positions of $[r_n+n, r_n+2n-1]$ into
$J_2$; further, there are two positions that are
$\Omega(C^3)$ positions apart, one in
$[r_n, r_n+n-1]$ and the other in $[r_n+n, r_n+2n-1]$,
such that $f_1$ maps the first position into $J_1$
and the second position into $J_2$. This implies that
$J_1$ must precede $J_2$, for otherwise Condition (b)
would be violated. Arguing similarly with $f_2$ in place
of $f_1$ (that is, the mapping from $\beta$ to $\xi$),
it follows that $J_2$ must precede $J_1$, a contradiction.   
We conclude that the string $\xi$ cannot exist.  
\end{proof}

Example \ref{exmp:reduction.separations} established separations
between various notions of recursive quasi-reducibility: pqi,
1qi and mqi-reducibilities. It remains to separate general 
quasi-isometry from its recursive counterpart.

\begin{theorem}
\label{thm:quasi.isom.non.recursive}
There exist two recursive strings $\alpha$ and $\beta$ such
that $\alpha \leq_{qi} \beta$ but $\alpha \not\leq_{mqi} \beta$.
\end{theorem}

\begin{proof}
We begin with an overview of the construction of $\alpha$ and
$\beta$. To ensure that only non-recursive 
quasi-isometries between $\alpha$ and $\beta$ exist,
we use a tool from computability theory, which is a 
Kleene tree \cite{Kleene52}---an infinite uniformly recursive
binary tree with no infinite recursive branches (see, for example, 
\cite[\S V.5]{Odifreddi89}).
The idea of the proof is to encode a fixed Kleene tree into $\beta$,
and construct $\alpha$ such that 
for any quasi-isometry $f$ from $\alpha$ to $\beta$,
an infinite branch of the encoded Kleene tree can be computed recursively
from $f$.
Hence, $f$ cannot be recursive, as otherwise
the chosen infinite branch of the Kleene tree must be recursive,
contradicting the definition of a Kleene tree.

We now describe the construction of $\alpha$ and $\beta$
based on some fixed Kleene tree $T \subseteq \{0, 1\}^*$.
The building blocks for $\alpha$ and $\beta$ are called {\em blocks},
or more specifically, {\em $n$-blocks} for some $n \in \bbbn$.
The construction will be done in stages,
where at stage $n$, we concatenate some $n$-blocks
to the existing prefixes of $\alpha$ and $\beta$.
An {\em $n$-block} is defined to be a string of one of the following forms:
\begin{align*}
\lam_{(n,0)} &= 0^n1^n, \\
\lam_{(n,i)} &= 0^{\lfloor \frac{n+1}{2} \rfloor}1^i
0^{\lceil \frac{n+1}{2} \rceil}1^n, \text{ for } 1 \leq i \leq n-1 \text{ or} \\
\lam'_{n} &= (01)^n1^n.
\end{align*}
The strings appended to $\alpha$ and $\beta$ at stage $n$ will be called $\theta_n$ and $\zeta_n$ respectively.
Taking the limit as $n$ grows to infinity, $\alpha$ and $\beta$ 
have the following shapes:
\begin{align*}
\alpha &= \theta_1\theta_2\cdots = (\theta_n)_{n=1}^{\infty}\,, \\
\beta &= \zeta_1\zeta_2\cdots = (\zeta_n)_{n=1}^{\infty}\,
\end{align*}
where $\theta_n$ and $\zeta_n$ are made up of the same number of $n$-blocks and
$|\zeta_n| \geq |\theta_n|$.
We now define $\theta_n$ and $\zeta_n$ for $n \in \bbbn$.
Set \[\theta_1 = \zeta_1 = \lam_{(1,0)}.\]
For $n \geq 2$, each of the strings $\theta_n$ and $\zeta_n$ is composed of
three main segments: a {\em scaling} segment, a {\em branching}
segment and a {\em selection} segment. These segments are added to 
$\theta_n$ and $\zeta_n$ in the given order,
each preceded by a {\em join} segment.
Furthermore, a scaling segment is further made up of two {\em scaling parts}
joined by a join segment.
So, for $n \geq 2$, the structure of $\theta_n$ can be depicted as follows:
\[
\theta_n = \underbrace{v_{n, 1}}_{\text{Join}} 
\quad \underbrace{s_{n, 1} \;\, v_{n, 2} \;\, s_{n, 2}}_{\text{Scaling}} 
\quad \underbrace{v_{n, 3}}_{\text{Join}}
\quad \underbrace{\vphantom{v_{n, 1}} t_{n}}_{\mathclap{\text{Branching}}} 
\quad \underbrace{v_{n, 4}}_{\text{Join}}
\quad \underbrace{\vphantom{v_{n, 1}} u_{n}}_{\mathclap{\text{Selection}}}
\]
and similarly,
\[
\zeta_n = \underbrace{v'_{n, 1}}_{\text{Join}} 
\quad \underbrace{s'_{n, 1} \;\, v'_{n, 2} \;\, s'_{n, 2}}_{\text{Scaling}} 
\quad \underbrace{v'_{n, 3}}_{\text{Join}}
\quad \underbrace{\vphantom{v'_{n, 1}} t'_{n}}_{\mathclap{\text{Branching}}} 
\quad \underbrace{v'_{n, 4}}_{\text{Join}}
\quad \underbrace{\vphantom{v'_{n, 1}} u'_{n}}_{\mathclap{\text{Selection}}}.
\]
We can now define each segment of $\theta_n$ and $\zeta_n$
and briefly explain its function,
and in the process give a high level overview of the proof.

\subsubparagraph{Join segment.}
Each join segment serves as a connector between
two different segments which aren't join segments.
A join segment in $\theta_n$ or $\zeta_n$ is also called an {\em $n$-join segment}.
For $i \in \{1, 2, 3, 4\}$, the $n$-join segments $v_{n, i}$ and $v'_{n, i}$ are defined as follows
\[v_{n, i} = v'_{n, i} = (\lambda_{(n, 0)})^{3nB^n_{2i - 1}}\]
where $B^n_{2i - 1}$ is
the number of blocks in $\alpha$ before the start of $v_{n,i}$.
We ensure that the corresponding segments of $\alpha$ and
$\beta$ have the same number of blocks.
So, $B^n_{2i - 1}$ is also the number of blocks in $\beta$ before the start of $v'_{n, i}$.

Given any fixed quasi-isometric reduction $f$ from $\alpha$ to $\beta$, we define the {\em lead} $\ell$ of an $n$-join segment $v_{n, i}$ such that for all $nB^n_{2i - 1} + 1 \leq j \leq 2nB^n_{2i - 1}$, $f$ maps 
the $j$-th $\lambda_{(n, 0)}$ block of $v_{n, i}$ to
the $(j + \ell)$-th $\lambda_{(n, 0)}$ block of $v'_{(n, i)}$.
We will show later that for large enough $n$,
the lead $\ell$ is always defined and non-negative.
To explain the functions of the other segments,
we will describe how each segment affects the leads of join segments next to it.

\subsubparagraph{Selection segment.}
The selection segment plays a key role in 
the encoding of the Kleene tree $T$ into the string $\beta$.
Before we define the selection segment, we first define
\[S_n = \left\{ \sum_{m = 1}^{n - 1} b_m4^{n - 1 - m} : 
b_1 \cdots b_{n - 1} \in T \cap \{0, 1\}^{n - 1} \right\}.\]
The set $S_n$ encodes all the strings of length $n-1$ in the Kleene tree $T$,
where each element $\sum_{m = 1}^{n - 1} b_m4^{n - 1 - m} \in S_n$ 
is the number with base-$4$ representation $b_1 \cdots b_{n - 1} \in T$ possibly with leading 0's.
We can now define the {\em selection segments} $u_n$ and $u'_n$. 
Define \[u_n = \lambda_{(n, 1)} (\lambda_{(n, 0)})^{\max(S_n)}\]
and for $1 \leq i \leq \max(S_n) + 1$, let the $i$-th block of $u'_n$ be:
\begin{itemize}
\item $\lambda_{(n, 0)}$ if $i - 1 \not\in S_n$ and
\item $\lambda_{(n, 1)}$ if $i - 1 \in S_n$.
\end{itemize}
Then, the selection segment $u'_n$ of $\beta$ encodes the set $S_n$,
which in turn encodes the set of all the strings in $T$ of length $n-1$.
So, all of the selection segments of $\beta$ together
encode the fixed Kleene tree $T$.
Meanwhile, each selection segment of $\alpha$
has a single $\lambda_{(n, 1)}$ block followed by $\lambda_{(n, 0)}$ blocks.
The single $\lambda_{(n, 1)}$ block serves as a pointer
which must be mapped by $f$ to a $\lambda_{(n, 1)}$ block 
in the respective selection segment of $\beta$.
Hence, the selection segments of $\alpha$ and $\beta$ 
ensure that for sufficiently large $n$,
the lead of a quasi-isometric reduction from $\alpha$ to
$\beta$ in the $n$-join segment preceding a selection segment is a number 
in $S_n$.
Moreover, the first $(n+1)$-join segment succeeds the $n$-th selection segment and has the same lead as the previous join segment.

The other segments ensure that the number in $S_n$ is chosen 
appropriately such that
an infinite branch of the Kleene tree
can be computed from the leads of a quasi-isometric reduction 
in the join segments preceding the selection segments.
More specifically, we need to make sure that 
for large enough $n$, 
the base-4 representation $b_1 \ldots b_{n-2}$ of 
the lead of the $(n-1)$-join segment preceding a selection segment
and the base-4 representation $b'_1 \ldots b'_{n-1}$ of
the lead of the $n$-join segment preceding a selection 
segment have the same first $n-c$ digits,
where $c$ is some constant independent of $n$.
That is, $b_1 \ldots b_{n-c} = b'_1 \ldots b'_{n-c}$.

\subsubparagraph{Scaling segment.}
To achieve the objective described above, each scaling segment helps 
by making sure that the lead of the join segment following the scaling segment
is 4 times that of the previous join segment.
To do this, a scaling segment is made up of two scaling parts
joined together by a join segment,
where each scaling part
doubles the lead of the previous join segment.
Then, the {\em scaling segments} of $\theta_n$ and $\zeta_n$ can be depicted as
$s_{n, 1} v_{n, 2} s_{n, 2}$ and 
$s'_{n, 1} v'_{n, 2} s'_{n, 2}$ respectively,
where $s_{n, 1}$, $s_{n, 2}$, $s'_{n, 1}$ and $s'_{n, 2}$ are
{\em scaling parts} defined as follows:
\[s_{n, i} = s'_{n, i} = 
(\lambda_{(n, 1)})^{nB^n_{2i}}(\lambda_{(n, 2)})^{nB^n_{2i}} \ldots 
(\lambda_{(n, n - 1)})^{nB^n_{2i}} (\lambda_{(n, 0)})^{2nB^n_{2i}}\]
where $i \in \{1, 2\}$ and $B^n_{2i}$ is the number of blocks in $\alpha$ before the start of $s_{n,i}$.
The doubling of the lead follows from the properties of the $n$-blocks chosen 
to make the scaling parts, which will be proven in the later parts of the proof.

\subsubparagraph{Branching segment.}
Note that the join segment \textit{after} a branching segment
precedes a selection segment and so must have a lead which is in $S_n$.
On the other hand, the join segment \textit{before} a branching segment
succeeds a scaling segment and may not be in $S_n$.
So, the branching segment's purpose is to allow minor adjustments to the lead
so that the lead after the branching segment is in $S_n$.
Furthermore, this adjustment must be small enough so that 
for large enough $n$, 
the first $n-c$ digits of the base-4 representations $b_1 \ldots b_{n-2}$ and $b'_1 \ldots b'_{n-1}$ of 
the leads of join segments $v_{n-1, 4}$ and $v_{n, 4}$ match,
where $c$ is some constant independent of $n$.
So, we define the {\em branching segments} $t_n$ and $t'_n$ as follows:
\begin{align*}  
t_n &= (\lambda_{(n, 0)})^{2nB^n_6 + 1} \text{ and}\\
t'_n &= (\lambda_{(n, 0)})^{2nB^n_6} \lambda'_n
\end{align*}
where $B^n_6$ is the number of blocks in $\alpha$ before the start of $t_n$.
The $n$-blocks chosen to make the branching segments $t_n$ and $t'_n$
ensure that given
a $C$-quasi-isometric reduction $f$ from $\alpha$ to $\beta$ and large enough $n$,
the lead of the join segment after the branching segment
is between $\ell-C$ and $\ell+1$ inclusive, 
where $\ell$ is the lead of the previous join segment.

The above descriptions give a high-level overview of the proof
while omitting the details of how certain properties are achieved.
These details will be given in the rest of the proof.

Note also that a recursive formula for the number of blocks
in each segment of $\theta_n$ or $\zeta_n$
may be determined in terms of $n$ and $\max(S_n)$, 
although, as the present proof does not analyze time or space 
complexity issues, such a formula will not be explicitly stated. 

We will now describe each segment in detail.
In each segment, we will restate the definition of the segment
and prove properties related to the segment.
We will occasionally be informal and speak of mappings
between two sequences of blocks; it is to be understood 
that in such a situation we are really referring to mappings 
between the sequences of positions of the block sequences in
question.

\proofsubparagraph{Join segment.} 
Recall that an {\em $n$-join segment} is
a sequence of $3nB$ blocks $\lam_{(n,0)}$,
where $B$ is the number of blocks in 
the prefix of $\alpha$ (resp.~$\beta$) just before the start
of the said sequence of $\lam_{(n,0)}$ blocks. The sequence 
of the $(nB+1)$-st to the $(2nB)$-th blocks of an $n$-join segment will 
be called an {\em $n$-inner join segment}.  
An $n$-join segment and an $n$-inner join segment
will be called a join segment and an inner join
segment respectively when the choice of $n$ is clear from the 
context.

The extension $\theta_n$ is defined so that for any quasi-isometric 
reduction $f$ from $\alpha$ to $\beta$, 
if $n$ is large enough and $\theta_n,\zeta_n$ each contains at
least $K$ join segments, then $f$ maps the sequence of positions of 
the $K$-th inner join segment in $\theta_n$ into a sequence of positions
of the $K$-th join segment of $\zeta_n$ in a {\em monotonic and 
one-block-to-one-block} fashion, by which we mean that there is a constant 
$t$ such that for $1 \leq i \leq nB$, $f$ maps the $i$-th 
$\lam_{(n,0)}$ block of the inner join segment in $\theta_n$ to the 
$(i+t)$-th $\lam_{(n,0)}$ block of the join segment of 
$\zeta_n$.
Furthermore, suppose that the first $\lam_{(n,0)}$ block of the $K$-th
inner join segment in $\theta_n$ is the $k_1$-st 
block of $\theta_1\cdots\theta_n$, and that the $(t+1)$-st $\lam_{(n,0)}$ 
block of the $K$-th join segment in $\zeta_n$ is the 
$k_2$-nd block of $\zeta_1\cdots\zeta_n$. 
Then, we call the quantity $k_2-k_1$ the {\em lead}
of $f$ in the sequence of positions of the $K$-th join segment of
$\theta_n$. Thus, the lead of a quasi-isometry from $\alpha$ to
$\beta$ is defined with respect to the sequence of positions
of a given $n$-join segment when $n$ is large enough. 

The extensions $\theta_n$ and $\zeta_n$ are chosen
so that when $n$ is large enough, the lead of $f$ in each 
sequence of positions of a join segment of $\theta_n$
is nonnegative.    
Moreover, the extensions are chosen so that 
there is a constant $C'$ (depending on $f$) such that 
when $n$ is large enough,  the lead $\ell$ of $f$ in the sequence of 
positions of the last join segment of $\theta_n$ is contained 
in $S_n$ and the string in $T \cap \{0, 1\}^{n-1}$ corresponding to $\ell$ has
a common prefix of length at least $n-C'$ with
the string in $T \cap \{0, 1\}^{n-2}$ corresponding to the analogously defined
lead at the end of stage $n-1$. The idea is that by calculating 
successive values of the lead, one could then compute recursively 
in $f$ an infinite branch of the Kleene tree.  

Based on the preliminarily defined shapes of $\alpha$ 
and $\beta$, we now state a few useful properties of 
quasi-isometric reductions from $\alpha$ to $\beta$, in 
particular how they map between various types of blocks when
the block lengths are large enough. Further details of
the construction will be provided progressively.

For $n \geq 2$, define $B_n$ to be the total number of blocks in 
$\zeta_1\cdots\zeta_{n-1}$.
The first observation is that for any $C$-quasi-isometric 
reduction, when $n$ is sufficiently large and $i > nB_n$, 
the $i$-th occurrence of an $n$-block in $\alpha$ 
cannot be mapped to an $m$-block in $\beta$ with $m < n$.  

\smallskip
\begin{claim}\label{clm:main.clm1}
Let $f$ be any $C$-quasi-isometric reduction
from $\alpha$ to $\beta$.
Then, for all sufficiently large $n$ and all $i > nB_n$, no 
position of the $i$-th occurrence of an $n$-block in $\alpha$ 
is mapped by $f$ to the position of an $m$-block in $\beta$
with $m < n$. 
\end{claim}

\smallskip
\begin{claimproof}
%Let $Q$ be the interval of positions of $\beta$ of all
%$m$-blocks with $m \leq C$. By Lemma \ref{lem:collision},
%$\max(f^{-1}(Q))$ is finite. Fix any
%$n > \max\{\max(f^{-1}(Q)),
Fix any $n > 4C+4$ and any $i > nB_n$, and let $p$ be a 
position of the $i$-th $n$-block in 
$\alpha$. There are at least $i-C-1 > nB_n-C-1$ $n$-blocks 
preceding position $p$ in $\alpha$ such that all the positions
of these $n$-blocks are more than $C$ positions away from
$p$. As $\lam_{(n,0)}$ is the shortest $n$-block and has a length
of $2n$, the total number of positions occupied by these blocks is
more than $2n(nB_n-C-1)$. By Condition (b), the
images of the positions of these $n$-blocks under $f$ precede $f(p)$. 
By Lemma \ref{lem:collision},
each position of $\beta$ has at most $C+1$ preimages,
which means that there are at least $\frac{2n(nB_n-C-1)}{C+1}$ 
positions preceding $f(p)$. In other words, 
$f(p) > \frac{2n(nB_n-C-1)}{C+1}$. Since each $k$-block with
$k < n$ has length at most $3(n-1)$ (a $\lam'_{n-1}$ block),
the number of $k$-blocks with $k < n$ that can fit
$\frac{2n(nB_n-C-1)}{C+1}$ positions must be at least
$\frac{2n(nB_n-C-1)}{3(n-1)(C+1)}$. By the choice of $n$,
\[
\begin{aligned}
\frac{2n(nB_n-C-1)}{3(n-1)(C+1)} &\geq \frac{2nB_n}{3(C+1)} - \frac{2n}{3(n-1)} \\
&> \frac{2(4C+4)B_n}{3(C+1)} - 2 \\
&\geq 2B_n - 2 \\
&\geq B_n. 
\end{aligned}
\]
Since there are exactly $B_n$ $k$-blocks with $k < n$, 
we conclude that $f(p)$ cannot occur in a $k$-block
with $k < n$.
\end{claimproof}

\smallskip
The next observation gives a localised restriction on
quasi-isometric mappings, in particular between 
$\theta_n$ and $\zeta_n$. 

\smallskip
\begin{claim}\label{clm:main.clm2}
Let $f$ be any $C$-quasi-isometric reduction
from $\alpha$ to $\beta$.
Then, for all sufficiently large $n$, no position of
a block $\lam_{(n,1)}$ is mapped by $f$ to a position
of a block $\lam_{(n,0)}$ occurring in $\beta$.
\end{claim}    

\smallskip
\begin{claimproof}
Pick any $n > 2C+2$. First, we observe that if at
least one position in the sequence $P$ of 
positions of $0^{\lfloor \frac{n+1}{2} \rfloor}$
is mapped to a position of some block $\lam_{(n,0)}$,
then $f$ maps the whole sequence $P$ into the sequence 
of the first $n$ positions of $\lam_{(n,0)}$. The reason 
is that if $f$ maps some position of $P$ to the position
of another block, then, since the $0$'s occurring in
the blocks adjacent to $\lam_{(n,0)}$ are at least
$n-1 > 2C+1$ positions away, Condition (a) would not
be satisfied. The same observation applies to the
sequence of positions of $0^{\lceil \frac{n+1}{2} \rceil}$.

Second, if the sequence of positions of 
$0^{\lfloor \frac{n+1}{2} \rfloor}$ is mapped into 
the sequence of positions of $0^n$ and the sequence of 
positions of $0^{\lceil \frac{n+1}{2} \rceil}$ is mapped
into the sequence of positions of $0$'s in the block
succeeding $\lam_{(n,0)}$, then the position of the
single $1$ in $\lam_{(n,1)}$ must be mapped 
into the sequence of positions of $n$ $1$'s in $\lam_{(n,0)}$.
But the position of this $1$ would then be at least
$\frac{n}{2} > C+1$ positions away from at least
one of the images of the $0$'s adjacent to the 
single $1$ occurring in $\lam_{(n,1)}$, contradicting
Condition (a).

Third, if the sequence of positions of 
$0^{\lfloor \frac{n+1}{2} \rfloor}$ and the sequence of 
positions of $0^{\lceil \frac{n+1}{2} \rceil}$
are both mapped into the sequence of positions of 
$0^n$ in $\lam_{(n,0)}$, then the image of at
least one position of $0^{\lceil \frac{n+1}{2} \rceil}$
that is at least $C+1$ positions after the
single $1$ in $\lam_{(n,1)}$ would precede the
image of the single $1$, contradicting Condition (b).
For a similar reason, the sequence of positions of
$0^{\lfloor \frac{n+1}{2} \rfloor}$ and the sequence of 
positions of $0^{\lceil \frac{n+1}{2} \rceil}$
cannot be both mapped into the sequence of positions of 
$0$'s in the succeeding block of $\lam_{(n,0)}$. 
\end{claimproof}

\smallskip
It is observed next that when $n$ is large enough, every
quasi-isometric mapping from an inner join segment 
of $\theta_n$ to a join segment of $\zeta_n$
is monotonic and one-block-to-one-block.
We recall that $B_n$ is the total number of blocks in 
$\zeta_1\cdots\zeta_{n-1}$.  
 
\smallskip
\begin{claim}\label{clm:main.clm3}
Let $f$ be any quasi-isometric reduction
from $\alpha$ to $\beta$. 
Then, for all sufficiently large $n$, if
$f$ maps a position of the $K$-th inner join segment 
of $\theta_n$ to a position of the $K'$-th join 
segment of $\zeta_n$, then there is a constant 
$t$ such that whenever $1 \leq i \leq nB_n$, $f$ maps the 
sequence of positions of the $i$-th $\lam_{(n,0)}$ block 
of the $K$-th inner join segment to the sequence of positions 
of the $(i+t)$-th $\lam_{(n,0)}$ block of the $K'$-th join 
segment of $\zeta_n$. 
\end{claim}
    
% \smallskip
\begin{claimproof}
Suppose that $f$ maps a position of the $i$-th $\lam_{(n,0)}$
block of the $K$-th inner join segment of $\theta_n$ to a 
position of the $j$-th $\lam_{(n,0)}$ block of the $K'$-th 
join segment of $\zeta_n$. Using a similar argument as in
the proof of Claim \ref{clm:main.clm2}, if $n$ is large enough,
then all the positions of the $i$-th $\lam_{(n,0)}$ block
must be mapped into the sequence of positions of the
$j$-th $\lam_{(n,0)}$ block, for otherwise Condition (a) would
fail. Inductively, assume that the $(i-k')$-th 
block is mapped to the $(j-k')$-th block and 
the $(i+k')$-th $\lam_{(n,0)}$ block is mapped to the
$(j+k')$-th block whenever $0 \leq k' \leq k$.
By Condition (b), when $n$ is large enough, $f$ cannot 
map the $(i-k-1)$-st block
to the $(j-k)$-th block or any subsequent block,
and $f$ also cannot map any block after the $(i-k)$-th
block to the $(j-k-1)$-st block.
If $f$ does not map the $(i-k-1)$-st block to the
$(j-k-1)$-st block, then the $(i-k-1)$-st block must
be mapped to some sequence of positions wholly before the
$(j-k-1)$-st block. But in this case, by Condition (b), 
no block before the $(i-k-1)$-st block can be mapped
to the $(j-k-1)$-st block, and so the $(j-k-1)$-st
$\lam_{(n,0)}$ block of the $K'$-th join segment of
$\zeta_n$ would have no preimage, which, in view of 
Lemma \ref{lem:image.point.density}, is false for
large enough values of $n$. Hence, the $(i-k-1)$-st
block must be mapped to the $(j-k-1)$-st block.
A similar argument shows that the $(i+k+1)$-st
block is mapped to the $(j+k+1)$-st block (for
large enough $n$).
\end{claimproof}

\smallskip
The scaling, branching and selection segments of 
$\theta_n$ and $\zeta_n$ will now be described in
detail. Along the way, we prove further properties of 
quasi-isometries from $\alpha$ to $\beta$.

\proofsubparagraph{Scaling segment.}
Set $B'_n = B_n + 3nB_n$, that is, $B'_n$ is the number of blocks in 
$\zeta_1\cdots\zeta_{n-1}$ plus the number of blocks in the
first join segment of $\zeta_n$. 
The scaling segments of $\theta_n$ and $\zeta_n$ are composed 
of two similar parts joined by an $n$-join segment.
The first part of the scaling segment of $\theta_n$ and of $\zeta_n$ is
\[
(\lam_{(n,1)})^{nB'_n}(\lam_{(n,2)})^{nB'_n}\cdots(\lam_{(n,n-1)})^{nB'_n}
(\lam_{(n,0)})^{2nB'_n}\,.
\]
We then append a join segment to the first part of the scaling 
segment of $\theta_n$ (resp.~$\zeta_n$). Let $B''_n$ be the total 
number of blocks in the prefix of $\alpha$ built so far. The 
second part of the scaling segment of $\theta_n$ and of $\zeta_n$ is 
\[
(\lam_{(n,1)})^{nB''_n}(\lam_{(n,2)})^{nB''_n}\cdots(\lam_{(n,n-1)})^{nB''_n}
(\lam_{(n,0)})^{2nB''_n}\,.
\]
The structures of the scaling segments of $\theta_n$ and $\zeta_n$
imply that when $n$ is large enough, the lead of a quasi-isometric
reduction in the interval of positions of the first inner join segment
of $\theta_n$ is nonnegative. This property will help to control 
the value of the lead in the last inner join segment of
$\theta_n$.  

\smallskip
\begin{claim}\label{clm:main.clm4}
Let $f$ be any $C$-quasi-isometric reduction
from $\alpha$ to $\beta$. 
Then, for all sufficiently large $n$, $f$ maps
the sequence of positions of the first inner join 
segment of $\theta_n$ to the sequence of positions
of the first join segment of $\zeta_n$ in a 
monotonic and one-block-to-one-block fashion.
Further, suppose that the first $\lam_{(n,0)}$
block of the first inner join segment of $\theta_n$ 
is the $(i+1)$-st block of $\theta_1\cdots\theta_n$.
Then, there is a nonnegative constant $\ell$, called
the {\em lead} of $f$ in the sequence of positions
of the first inner join segment of $\theta_n$,
such that whenever $1 \leq k \leq nB_n$, 
$f$ maps the sequence of positions of the $(i+k)$-th block 
of $\theta_1\cdots\theta_n$ into the sequence of
positions of the $(i+k+\ell)$-th block of
$\zeta_1\cdots\zeta_n$.
\end{claim}
    
\smallskip
\begin{claimproof}
We first show that the first $\lam_{(n,0)}$ block
of the first inner join segment of $\theta_n$
is mapped to some $\lam_{(n,0)}$ block of the first
join segment of $\zeta_n$. By Claim \ref{clm:main.clm1},
$\lam_{(n,0)}$ cannot be mapped to any $k$-block
with $k < n$ for large enough $n$. By the definition of 
$B_n$, for large enough $n$, the positions of 
$\theta_1\cdots\theta_{n-1}$ are mapped to at most
$nB_n$ blocks in the first join segment of $\zeta_n$.
By Claim \ref{clm:main.clm3}, the first $nB_n$ blocks
$\lam_{(n,0)}$ of the first join segment of $\theta_n$
are mapped to at most $nB_n$ blocks in the first
join segment of $\zeta_n$. Thus, there are at
least $nB_n$ blocks $\lam_{(n,0)}$ of the first
join segment of $\zeta_n$ that have preimages 
after the first $nB_n$ blocks $\lam_{(n,0)}$
of the first join segment of $\theta_n$.
By Condition (b), the first $\lam_{(n,0)}$ block 
of the first inner join segment of $\theta_n$ must
therefore be mapped to some $\lam_{(n,0)}$ block in 
the first join segment of $\zeta_n$. 

Now we show that the first $\lam_{(n,0)}$ block
of the first inner join segment of $\theta_n$
is mapped to a $\lam_{(n,0)}$ block after the
first $nB_n$ blocks of the first join segment
of $\zeta_n$. Since, for each $i \leq n-1$, 
the number of blocks in $\zeta_i$ is equal
to the number of blocks in $\theta_i$,
this would imply that the lead of $f$ in the
first inner join segment of $\theta_n$ is
nonnegative. By Claim \ref{clm:main.clm3},
for large enough $n$,
if the first block of the first inner join
segment of $\theta_n$ is mapped to one of the 
first $nB_n$ blocks of the first join segment of
$\zeta_n$, then the sequence of positions of at 
least one $\lam_{(n,0)}$ block of the first join 
segment of $\zeta_n$ has a preimage after the first 
join segment of $\theta_n$. By Condition (b),
if $n > C^2$, then this preimage is included in the
sequence of positions of $(\lam_{(n,1)})^{nB'_n}$,
a prefix of the scaling segment. But by Claim 
\ref{clm:main.clm2}, $\lam_{(n,1)}$ cannot be mapped 
to $\lam_{(n,0)}$. Hence, the first $\lam_{(n,0)}$ block
of the first inner join segment of $\theta_n$
is mapped to a $\lam_{(n,0)}$ block after the
first $nB_n$ blocks of the first join segment
of $\zeta_n$.
\end{claimproof}

\smallskip
After the first part of the scaling segment,
when $n$ is large enough, the lead of $f$ in the 
succeeding inner join segment is double the lead 
of $f$ in the preceding inner join segment; 
after the second part, the lead is quadrupled. 
To see this, we first observe that within a scaling 
segment, when $n$ is large enough and 
$i \leq n-2$, any quasi-isometric reduction must map 
a single $\lam_{(n,i)}$ block to a single $\lam_{(n,i)}$ 
block or to a single $\lam_{(n,i+1)}$ block.

\smallskip
\begin{claim}\label{clm:main.clm5}
Let $f$ be any $C$-quasi-isometric reduction
from $\alpha$ to $\beta$. 
Then, for large enough $n$ and $i \leq n-2$, 
$f$ maps the sequence of positions of each 
$\lam_{(n,i)}$ block in a scaling segment of 
$\theta_n$ into either the sequence of positions 
of a $\lam_{(n,i)}$ block in a scaling segment 
of $\zeta_n$ or the sequence of positions of a 
$\lam_{(n,i+1)}$ block in a scaling segment of 
$\zeta_n$.
\end{claim}

\smallskip
\begin{claimproof}
According to Claim \ref{clm:main.clm4}, when $n$ is
large enough, the lead of $f$ in the preceding inner
join segment is nonnegative. In this case, none of
the blocks in the scaling segment of $\theta_n$ are
mapped to blocks before the scaling segment of 
$\zeta_n$. Furthermore, if $n > 4C+4$, then 
the positions of at most $nB'_n$ blocks in the scaling 
segment of $\zeta_n$ have preimages before the
scaling segment of $\theta_n$. By Condition (b),
only the first $nB'_n$ blocks of the scaling segment
can have preimages before the scaling segment of 
$\theta_n$. 

Now consider a $\lam_{(n,i)}$ block for some $i \leq n-2$.
When $n$ is large and $j \leq n-2$, the positions of at 
most one $\lam_{(n,i)}$ block can be mapped to a 
$\lam_{(n,j)}$ block.
To see this, we note the following subclaim.

\smallskip
\begin{subclaim}\label{subclm:lami}
Suppose that $n$ is large enough so that the 
$\lam_{(n,0)}$ block just before the first part of the 
scaling segment of $\theta_n$ is mapped into the
$t$-th block of $\zeta_1\cdots\zeta_n$, where this
$t$-th block is either the $\lam_{(n,0)}$ block just 
before the first part of the scaling segment of 
$\zeta_n$ or a $\lam_{(n,1)}$ block in the first part 
of the scaling segment of $\zeta_n$. Then, for all $s$
such that the $s$-th block of the first part of the scaling 
segment of $\theta_n$ is some $\lam_{(n,i)}$ block with 
$i \leq n-2$, this $s$-th block is mapped into the 
$(t+s)$-th block of $\zeta_1\cdots\zeta_n$.
\end{subclaim}

\smallskip
\begin{claimproof}[Proof of Subclaim \ref{subclm:lami}.]
By the choice of $B'_n$ and for large enough $n$, 
the $(t+1)$-st block of $\zeta_1\cdots\zeta_n$ is either
a $\lam_{(n,1)}$ block or a $\lam_{(n,2)}$ block.
Since a $\lam_{(n,1)}$ block cannot be mapped into a
$\lam_{(n,0)}$ block, some parts of the first 
$\lam_{(n,1)}$ block in the scaling segment must be mapped
into the $(t+1)$-st block of $\zeta_1\cdots\zeta_n$
in order to avoid causing large gaps in the range of $f$. 
However, the first $\lam_{(n,1)}$ block in the scaling 
segment cannot be mapped across two adjacent blocks of the 
shape $\lam_{(n,1)}\lam_{(n,1)}$ or
$\lam_{(n,1)}\lam_{(n,2)}$. For,
if the single $1$ were mapped into $1$ in the first 
$\lam_{(n,1)}$ block, then the succeeding
substring $0^{\lceil \frac{n+1}{2} \rceil}1^n$ would have to
be mapped into $0^{\lceil \frac{n+1}{2} \rceil}1^n$ in order
to avoid causing large gaps in the range of $f$.
For the same reason, the single $1$ cannot be mapped to 
$1^n$. Similarly, the single $1$ cannot be mapped
into $1$ in the second $\lam_{(n,1)}$ block or into
$1^2$ in the $\lam_{(n,2)}$ block (for large enough $n$).
The suffix $1^n$ of $\lam_{(n,1)}$ would then have to be 
mapped into the suffix $1^n$ of $\lam_{(n,1)}$.
It follows that the first $\lam_{(n,1)}$ block of the 
scaling segment must be wholly mapped into the 
$(t+1)$-st block of $\zeta_1\cdots\zeta_n$.

Arguing as before, the next $\lam_{(n,1)}$ or 
$\lam_{(n,2)}$ block of $\theta_n$ cannot be mapped across
two adjacent blocks of the shape $\lam_{(n,1)}\lam_{(n,1)}$,
$\lam_{(n,1)}\lam_{(n,2)}$ or $\lam_{(n,2)}\lam_{(n,2)}$.
By Condition (b), the next $\lam_{(n,1)}$ or $\lam_{(n,2)}$ 
block also cannot be entirely mapped into the $(t+1)$-st 
block of $\zeta_1\cdots\zeta_n$. Thus, the next $\lam_{(n,1)}$
or $\lam_{(n,2)}$ block must be wholly mapped into the
$(t+2)$-nd block of $\zeta_1\cdots\zeta_n$.   
Applying the preceding arguments inductively to the next
$\lam_{(n,1)}$ or $\lam_{(n,2)}$ block and then to subsequent 
$\lam_{(n,i)}$ blocks, it follows that there cannot be two 
adjacent halves of $\lam_{(n,i)}$ blocks or adjacent halves 
of a $\lam_{(n,i)}$ and a $\lam_{(n,i+1)}$ block that are 
mapped across a single $\lam_{(n,j)}$ block for each 
$j \leq n-2$. Furthermore, for $i \leq n-2$, each $s$-th 
subsequent $\lam_{(n,i)}$ block 
must be entirely mapped into the $(t+s)$-th block of 
$\zeta_1\cdots\zeta_n$, which is either a $\lam_{(n,i)}$ block 
or a $\lam_{(n,i+1)}$ block.
~\claimqedhere~(Subclaim \ref{subclm:lami}) 
\end{claimproof}

By Condition (b), $f$ cannot
map a $\lam_{(n,i)}$ block in the scaling segment
to any $\lam_{(n,j)}$ block in the same scaling segment
with $j > i+1$. The positions of a $\lam_{(n,i)}$ block
can be mapped in a one-to-one fashion to the positions
of another $\lam_{(n,i)}$ block. The positions of a 
$\lam_{(n,i)}$ block can also be mapped
one-to-one to the positions of a $\lam_{(n,i+1)}$ block,
as shown below.
\[
\begin{tikzpicture}
\node at (0,0) (lami01) {$0^{\lfloor \frac{n+1}{2} \rfloor}$};
\node at (0.8,0) (lami1i) {$1^i$};
\node at (1.5,0) (lami02) {$0^{\lceil \frac{n+1}{2} \rceil}$}; 
\node at (2.2,0) (lami1s) {$1^n$};
\node at (0,-1) (lamip101) {$0^{\lfloor \frac{n+1}{2} \rfloor}$};
\node at (0.8,-1) (lamip11i) {$1^{i+1}$};
\node at (1.7,-1) (lamip102) {$0^{\lceil \frac{n+1}{2} \rceil}$}; 
\node at (2.4,-1) (lamip11s) {$1^n$};

\draw[->,thick]
(lami01) -- (lamip101);
\draw[->,thick]
(lami1i) -- (lamip11i);
\draw[->,thick]
(lami02) -- (lamip102);
\draw[->,thick]
(lami1s) -- (lamip11s);
\end{tikzpicture}
\qedhere
\]
\end{claimproof}

\smallskip
The next claim shows that to achieve a doubling of
the lead, a single $\lam_{(n,n-1)}$ block can be mapped to
two adjacent $\lam_{(n,0)}$ blocks.   

\smallskip
\begin{claim}\label{clm:main.clm6}
Let $f$ be any $C$-quasi-isometric reduction
from $\alpha$ to $\beta$. 
Then, for large enough $n$, $f$ can map the
sequence of positions of a $\lam_{(n,n-1)}$ block 
into the sequence of positions of exactly $k$ 
blocks $\lam_{(n,0)}$ iff $k=2$. 
\end{claim}

\smallskip
\begin{claimproof}
If a $\lam_{(n,n-1)}$ block were mapped to a single 
$\lam_{(n,0)}$ block, then the positions of the 
substrings $0^{\lfloor \frac{n+1}{2} \rfloor}$
and $0^{\lceil \frac{n+1}{2} \rceil}$ must be mapped  
into the sequence of positions of $0^n$. But the 
positions of the substring $1^{n-1}$ would have to be 
mapped into the sequence of positions of $1^n$, and this 
would violate Condition (b) for large enough $n$. 
Furthermore, when $n > C+1$, a $\lam_{(n,n-1)}$ block 
cannot be mapped across more than two $\lam_{(n,0)}$ 
blocks without resulting in an interval of at least
$C+1$ positions of $\zeta_n$ having no preimage, 
contradicting Lemma \ref{lem:image.point.density}.    
A one-to-one mapping of the positions of a single 
$\lam_{(n,n-1)}$ block to the positions of two 
$\lam_{(n,0)}$ blocks is depicted in the next figure.

\[
\begin{tikzpicture}
\node at (0,0) (lamn01) {$0^{\lfloor \frac{n+1}{2} \rfloor}$};
\node at (1,0) (lamn1n) {$1^{n-1}$};
\node at (1.9,0) (lamn02) {$0^{\lceil \frac{n+1}{2} \rceil}$}; 
\node at (2.6,0) (lamn1s) {$1^n$};
\node at (0,-1) (lam010s) {$0^{n}$};
\node at (0.8,-1) (lam011s) {$1^{n}$};
\node at (1.7,-1) (lam020s) {$0^{n}$}; 
\node at (2.4,-1) (lam021s) {$1^{n}$};

\draw[->,thick]
(lamn01) -- (lam010s);
\draw[->,thick]
(lamn1n) -- (lam011s);
\draw[->,thick]
(lamn02) -- (lam020s);
\draw[->,thick]
(lamn1s) -- (lam021s);
\end{tikzpicture}
\]
Thus, when $n$ is large enough, a single $\lam_{(n,n-1)}$
block maps to exactly two $\lam_{(n,0)}$ blocks.
\end{claimproof}

\smallskip
We elaborate on why, for large enough $n$, the 
lead of $f$ is quadrupled at the end of the scaling segment. 
Suppose that $f$ has a nonnegative lead
$\ell$ at the start of the scaling segment. When $n$ is large
enough, $\ell \leq nB'_n$. By Claim \ref{clm:main.clm5},
for $1 \leq i \leq n-2$, the first $nB'_n-\ell$ blocks 
$\lam_{(n,i)}$ of the (first part of the) scaling segment of $\theta_n$ 
are mapped to the last $nB'_n-\ell$ blocks $\lam_{(n,i)}$ of the
(first part of the) scaling segment of $\zeta_n$, while the last 
$\ell$ blocks $\lam_{(n,i)}$
are mapped to the first $\ell$ blocks $\lam_{(n,i+1)}$.
Further, the first $nB'_n-\ell$ blocks $\lam_{(n,n-1)}$ are 
mapped to the last $nB'_n-\ell$ blocks $\lam_{(n,n-1)}$.
By Claim \ref{clm:main.clm6}, each of the last $\ell$ 
$\lam_{(n,n-1)}$ blocks must be mapped to two $\lam_{(n,0)}$
blocks, while the first $nB'_n-2\ell$ blocks $\lam_{(n,0)}$
of the (first part of the) scaling segment of $\theta_n$ are 
mapped in a monotonic and one-block-to-one-block fashion to the
remaining $nB'_n-2\ell$ blocks $\lam_{(n,0)}$.      
Thus, if $f$ has a lead of $\ell$ in the inner join segment 
preceding the current scaling segment, then it has a lead of 
$2\ell$ at the end of the first part of the scaling segment; 
after the second part of the scaling segment, a similar argument 
as before shows that the lead is quadrupled to $4\ell$. 
We next consider the branching segment.

\proofsubparagraph{Branching segment.} In this segment, the lead
of a $C$-quasi isometry $f$ is increased by $1$ or decreased 
by up to $C$ when $n$ is large enough. To achieve
this, suppose that, after the join segment of $\theta_n$ 
succeeding the scaling segment, there are altogether
$B$ blocks in $\zeta_1\cdots\zeta_{n-1}$ and the current
prefix of $\zeta_n$. The branching segment of $\theta_n$
is 
\[
(\lam_{(n,0)})^{2Bn+1}
\]
while the branching segment of $\zeta_n$ is
\[
(\lam_{(n,0)})^{2Bn}\lam'_n\,.
\] 

\smallskip
\begin{claim}\label{clm:main.clm7}
Let $f$ be any $C$-quasi-isometric reduction
from $\alpha$ to $\beta$. Suppose that $n$
is large enough so that the lead $\ell$
of $f$ in the sequence of positions of the 
join segment just before the branching segment
is defined and nonnegative. 
Then, for large enough $n$, the lead of $f$ in 
the sequence of positions of the join segment after
the branching segment is at least $\ell-C$ and
at most $\ell+1$. 
\end{claim} 

\smallskip
\begin{claimproof}
When $n > C$, the lead of $f$ in the sequence of 
positions of the join segment just before the branching 
segment is at most $nB$, where $B$ is the number of blocks
in the prefix of $\alpha$ (or $\beta$) built so far.
Thus, there are at most $nB$ blocks $\lam_{(n,0)}$ of the
branching segment of $\zeta_n$ that have preimages before
the start of the branching segment of $\theta_n$. 
Consequently, the first $\lam_{(n,0)}$ block of the branching 
segment of $\theta_n$ is mapped into the sequence of positions 
of the first $nB+1$ blocks $\lam_{(n,0)}$ of the branching 
segment of $\zeta_n$.

By Lemma \ref{lem:collision}, up to $C+1$
blocks $\lam_{(n,0)}$ can be mapped to a single
$\lam'_{n}$ block. Hence, the lead of $f$ from the 
previous inner join segment can be decreased by up 
to $C$ in the subsequent inner join segment.
Further, when $n$ is large enough,
a single $\lam_{(n,0)}$ block can be mapped across 
$\lam_{(n,0)}\lam'_n$. This can be done by mapping
the $0^n$ substring of the $\lam_{(n,0)}$ block in 
the branching segment of $\theta_n$ to the $0^n$ 
substring of the $\lam_{(n,0)}$ block in the branching
segment of $\zeta_n$, and then mapping $1^n$ to
$1^n(01)^n1^n$ by mapping the first third of $1$'s
to the first occurrence of $1^n$, the second third of 
$1$'s to $(01)^n$ and the last third of $1$'s to the 
second occurrence of $1^n$. The mapping is illustrated
as follows.
\[
\begin{tikzpicture}
\node at (0,0) (lamn0s) {$0^{n}$};
\node at (0.5,0) (lamn1s) {$1^{n}$};
\node at (0,-1) (lamn20s) {$0^{n}$}; 
\node at (0.5,-1) (lamn21s) {$1^n$};
\node at (1.2,-1) (lamnp1) {$(01)^{n}$};
\node at (1.9,-1) (lamnp2) {$1^{n}$};
\
\draw[->,thick]
(lamn0s) -- (lamn20s);
\draw[->,thick]
(lamn1s) -- (lamn21s);
\draw[->,thick]
(lamn1s) -- (lamnp1);
\draw[->,thick]
(lamn1s) -- (lamnp2);
\end{tikzpicture}
\]
The proof of Claim \ref{clm:main.clm3} shows that for
large enough $n$, one and only one $\lam_{(n,0)}$ block 
can be mapped to a $\lam_{(n,0)}$ block.
Thus, the remaining $\lam_{(n,0)}$ blocks in the branching
segment of $\theta_n$ are mapped in a monotonic
and one-block-to-one-block fashion to the rest of the 
$\lam_{(n,0)}$ blocks in the branching segment of 
$\zeta_n$.
\end{claimproof}

\smallskip
After appending an $n$-join segment to the branching segment,
the selection segment is the final main segment to be added.

\proofsubparagraph{Selection segment.} This segment filters out 
quasi-isometries whose lead in the previous inner join segment
is not an element of $S_n$. The selection segment for
$\theta_n$ is
\[
\lam_{(n,1)}(\lam_{(n,0)})^{\max(S_n)}
\] 
while the selection segment for $\zeta_n$ is a concatenation
\[
\lam_{(n,0)}\cdots\lam_{(n,1)}\cdots\lam_{(n,0)}\cdots
\lam_{(n,0)}\lam_{(n,1)}
\] 
of $\max(S_n)+1$ $n$-blocks such that the $i$-th block is
$\lam_{(n,0)}$ if $i-1 \notin S_n$ and is $\lam_{(n,1)}$ if
$i-1 \in S_n$. 

\smallskip
\begin{claim}\label{clm:main.clm8}
Let $f$ be any $C$-quasi-isometric reduction
from $\alpha$ to $\beta$. Suppose that $n$
is large enough so that the lead $\ell$
of $f$ in the sequence of positions of the 
join segment just before the selection segment
is defined and nonnegative. 
Then, for large enough $n$, $f$ maps the
sequence of positions of the $\lam_{(n,1)}$ 
block in the selection segment of $\theta_n$
into the sequence of positions of exactly
one of the $\lam_{(n,1)}$ blocks in the 
selection segment of $\zeta_n$. In particular,
$\ell \in S_n$.
\end{claim} 

\smallskip
\begin{claimproof}
By Claim \ref{clm:main.clm2}, for large enough $n$, 
$f$ cannot map the $\lam_{(n,1)}$ block of the 
selection segment of
$\theta_n$ to any position in the first join segment
of $\zeta_{n+1}$. Thus, $\ell \leq \max(S_n)$. By
Claim \ref{clm:main.clm2} again, $f$ must map the
$\lam_{(n,1)}$ block of the selection segment
of $\theta_n$ to some $\lam_{(n,1)}$ block in the
selection segment of $\zeta_n$, say the $c$-th block
of the selection segment of $\zeta_n$, for large 
enough $n$. By Condition (b), if $n > C$, then $f$ 
does not map any $\lam_{(n,0)}$ block of the selection 
segment of $\theta_n$ to a position before the image of the 
$\lam_{(n,1)}$ block under $f$. This implies that
$\ell = c-1 \in S_n$.
\end{claimproof}

% \smallskip
Putting everything together, the structures of $\theta_n$ and 
$\zeta_n$ for $n \geq 2$ look as follows:
\begin{equation}
\begin{aligned}\label{eqn:theta.zeta}
\theta_n 
&= v_{n, 1} \;\, s_{n, 1} \;\, v_{n, 2} \;\, s_{n, 2} \;\, 
v_{n, 3} \;\, t_n \;\, v_{n, 4} \;\, u_n\\
&= (\lam_{(n,0)})^{3nB^n_1} 
&&\leftarrow v_{n, 1}  &&\text{ (Join) }\\
&~~~~ (\lam_{(n,1)})^{nB^n_2}(\lam_{(n,2)})^{nB^n_2}\cdots(\lam_{(n,n-1)})^{nB^n_2}(\lam_{(n,0)})^{2nB^n_2} 
&&\leftarrow s_{n, 1}  &&\text{ (Scaling) }\\
&~~~~ (\lam_{(n,0)})^{3nB^n_3} 
&&\leftarrow v_{n, 2}  &&\text{ (Join) }\\ 
&~~~~ (\lam_{(n,1)})^{nB^n_4}(\lam_{(n,2)})^{nB^n_4}\cdots(\lam_{(n,n-1)})^{nB^n_4}
(\lam_{(n,0)})^{2nB^n_4} 
&&\leftarrow s_{n, 2}  &&\text{ (Scaling) }\\ 
&~~~~ (\lam_{(n,0)})^{3nB^n_5} 
&&\leftarrow v_{n, 3}   &&\text{ (Join) }\\ 
&~~~~ (\lam_{(n,0)})^{2nB^n_6+1} 
&&\leftarrow t_n  &&\text{ (Branching) }\\ 
&~~~~ (\lam_{(n,0)})^{3nB^n_7} 
&&\leftarrow v_{n, 4}  &&\text{ (Join) }\\ 
&~~~~ \lam_{(n,1)}(\lam_{(n,0)})^{\max(S_n)} 
&&\leftarrow u_n   &&\text{ (Selection) } \\ 
\\
\zeta_n 
&= v'_{n, 1} \;\, s'_{n, 1} \;\, v'_{n, 2} \;\, s'_{n, 2} \;\, 
v'_{n, 3} \;\, t'_n \;\, v'_{n, 4} \;\, u'_n\\
&= (\lam_{(n,0)})^{3nB^n_1} 
&&\leftarrow v'_{n, 1}  &&\text{ (Join) }\\
&~~~~ (\lam_{(n,1)})^{nB^n_2}(\lam_{(n,2)})^{nB^n_2}\cdots(\lam_{(n,n-1)})^{nB^n_2}(\lam_{(n,0)})^{2nB^n_2} 
&&\leftarrow s'_{n, 1}  &&\text{ (Scaling) }\\
&~~~~ (\lam_{(n,0)})^{3nB^n_3} 
&&\leftarrow v'_{n, 2}  &&\text{ (Join) }\\ 
&~~~~ (\lam_{(n,1)})^{nB^n_4}(\lam_{(n,2)})^{nB^n_4}\cdots(\lam_{(n,n-1)})^{nB^n_4}
(\lam_{(n,0)})^{2nB^n_4} 
&&\leftarrow s'_{n, 2}  &&\text{ (Scaling) }\\ 
&~~~~ (\lam_{(n,0)})^{3nB^n_5} 
&&\leftarrow v'_{n, 3}  &&\text{ (Join) }\\ 
&~~~~ (\lam_{(n,0)})^{2nB^n_6}\lam'_n 
&&\leftarrow t'_n  &&\text{ (Branching) }\\ 
&~~~~ (\lam_{(n,0)})^{3nB^n_7} 
&&\leftarrow v'_{n, 4}  &&\text{ (Join) }\\ 
&~~~~ \lam_{(n,0)}\cdots\lam_{(n,1)}\cdots\lam_{(n,0)}\cdots
\lam_{(n,0)}\lam_{(n,1)} 
&&\leftarrow u'_n   &&\text{ (Selection) } 
\end{aligned}
\end{equation}
where $B^n_i$ is the number of blocks in $\alpha$ 
(or $\beta$) preceding the segment where the 
parameter $B^n_i$ is first used.

Suppose that $f$ is a $C$-quasi-isometry from $\alpha$
to $\beta$. As explained earlier, if $n$ is sufficiently
large, then the lead $\ell$ of $f$ in the first join segment
of $\theta_{n+1}$ is contained in $S_{n}$. The lead is then 
quadrupled after the scaling segment. Further, by Claim 
\ref{clm:main.clm7}, the lead $\ell'$ of $f$ in the final join 
segment of $\theta_{n+1}$ is at least $4\ell-C$, at most $4\ell+1$ 
and is contained in $S_{n+1}$. To establish that 
$\alpha \not\leq_{mqi} \beta$, we prove that there is a 
constant $c$ such that the coefficients of the expressions for 
$\ell$ and $\ell'$ as linear combinations of powers of $4$
agree except on the smallest $c$ powers of $4$. Thus, the
branch on the Kleene tree represented by $\ell$ can be
properly extended by a branch represented by a prefix of
$\ell'$, so by repeatedly calculating the leads of $f$ for 
increasing values of $n$, one may construct an infinite branch 
of the tree. 

\smallskip
\begin{claim}\label{clm:main.clm9}
There exists a nonnegative constant $c$ such that if $n > c$,
$\ell = \sum_{m=1}^n b_m\cdot 4^{n-m} \in S_{n+1}$
and $\ell' = \sum_{m=1}^{n+1} b'_m\cdot 4^{n+1-m} \in S_{n+2}$,
where $b_m \in \{0,1\}$ and $b'_m \in \{0,1\}$, then
for all $m \in \{1,\ldots,n-c\}$, $b_{m} = b'_{m}$. 
\end{claim}

\smallskip
\begin{claimproof}
Let $c$ be a positive constant such that
$C < 4^c$. We show that for all 
$m \in \{1,\ldots,n-c\}$, $b_m = b'_m$.
As shown earlier, $4\ell-C \leq \ell' \leq 4\ell+1$.
Substituting the expressions for $\ell$ and $\ell'$
as linear combinations of powers of $4$ into the  
inequality $4\ell-C \leq \ell'$, 
\[
\begin{aligned}
\sum_{m=1}^{n} b_m4^{n-m+1} - 4^c < \sum_{m=1}^{n} b_m4^{n-m+1} - C
\leq \sum_{m=1}^{n+1} b_m'4^{n-m+1}\,.
\end{aligned}
\] 
Dividing both sides of the inequality by $4^{c+1}$,
\[
\begin{aligned}
\sum_{m=1}^{n} b_m4^{n-m-c} - \frac{1}{4} < \sum_{m=1}^{n+1} b'_m4^{n-m-c}\,.
\end{aligned}
\] 
Splitting the sums on both sides according to whether the powers of $4$
are nonnegative or negative,
\[
\begin{aligned}
\sum_{m=1}^{n-c} b_m4^{n-m-c} + \sum_{m=n-c+1}^{n} b_m4^{n-m-c} - 
\frac{1}{4} < \sum_{m=1}^{n-c} b'_m4^{n-m-c} + \sum_{m=n-c+1}^{n+1} b'_m4^{n-m-c}\,.
\end{aligned}
\] 
Rearranging,
\[
\begin{aligned}
\sum_{m=1}^{n-c} b_m4^{n-m-c} + \sum_{m=n-c+1}^{n} b_m4^{n-m-c} - 
\frac{1}{4} - \sum_{m=n-c+1}^{n+1} b'_m4^{n-m-c} < \sum_{m=1}^{n-c} b'_m4^{n-m-c}\,.
\end{aligned}
\] 
Since 
\[
\begin{aligned}
\sum_{m=n-c+1}^{n} b_m4^{n-m-c} - \frac{1}{4} - \sum_{m=n-c+1}^{n+1} b'_m4^{n-m-c}
\geq -\frac{1}{4} - \sum_{m=1}^{\infty} 4^{-m} = -\frac{7}{12}\,,
\end{aligned}
\]
it follows that
\[
\begin{aligned}
\sum_{m=1}^{n-c} b'_m4^{n-m-c} > \sum_{m=1}^{n-c} b_m4^{n-m-c} - \frac{7}{12}\,.
\end{aligned}
\]
Taking the ceiling on both sides,
\begin{align}
\sum_{m=1}^{n-c} b'_m4^{n-m-c} &= \left\lceil \sum_{m=1}^{n-c} b'_m4^{n-m-c} \right\rceil \nonumber\\
&\geq \left\lceil \sum_{m=1}^{n-c} b_m4^{n-m-c} - \frac{7}{12} \right\rceil \nonumber\\ 
&= \sum_{m=1}^{n-c} b_m4^{n-m-c} \label{eqn:constant.sep1}\,. 
\end{align}
From the inequality $\ell' \leq 4\ell+1$, 
\[
\begin{aligned}
\sum_{m=1}^{n-c} b'_m4^{n-m+1} + \sum_{m=n-c+1}^{n+1} b'_m4^{n-m+1}
\leq \sum_{m=1}^{n} b_m4^{n-m+1} + 1\,.
\end{aligned}
\]
Dividing both sides by $4^{c+1}$ and applying the floor function,

\begin{align}
\sum_{m=1}^{n-c} b'_m4^{n-m-c} &= 
\left\lfloor\sum_{m=1}^{n-c} b'_m4^{n-m-c} + \sum_{m=n-c+1}^{n+1} b'_m4^{n-m-c}\right\rfloor \nonumber\\
&\leq \left\lfloor\sum_{m=1}^{n} b_m4^{n-m-c} + \frac{1}{4^{c+1}} \right\rfloor \nonumber\\
&= \sum_{m=1}^{n-c} b_m4^{n-m-c} \label{eqn:constant.sep2}\,.
\end{align}
Inequalities (\ref{eqn:constant.sep1}) and (\ref{eqn:constant.sep2}) together
give that for each $m \in \{1,\ldots,n-c\}$, $b'_m = b_m$. 
\end{claimproof}

\smallskip
\begin{claim}\label{clm:main.clm10}
$\alpha \not\leq_{mqi} \beta$.
\end{claim}

\smallskip
\begin{claimproof}
Suppose that $f$ is a $C$-quasi-isometry from $\alpha$
to $\beta$. We show that an infinite branch of the 
Kleene tree can be determined recursively in $f$. 
Let $c$ be the constant in the statement
of Claim \ref{clm:main.clm9}. Fix $n_0 > c$ large enough
so that whenever $n \geq n_0$, $f$ has nonnegative lead 
in the last join segment of $\theta_{n}$. 
Then, calculate the sequence $(\ell_n)_{n=n_0}^{\infty}$
of leads of $f$ in the last join segment of $\theta_n$
for $n \geq n_0$, expressing each $\ell_n$ in the shape
$\sum_{m=1}^n b^n_m4^{n-m}$, where $b^n_m \in \{0,1\}$. 
By Claim \ref{clm:main.clm9},
$b^{n}_m = b^{n+1}_m$ for all $m \in \{1,\ldots,n-c\}$,
and so by Claim \ref{clm:main.clm8},
\[
b^{n_0}_1\cdots b^{n_0}_{n_0-c} b^{n_0+1}_{n_0-c+1}
b^{n_0+2}_{n_0-c+2} \cdots b^{n_0+k}_{n_0-c+k} \cdots
\]
is an infinite branch of the Kleene tree. But since the 
tree has no infinite recursive branches, $f$ must be 
nonrecursive.
\end{claimproof}

\smallskip
To finish the proof, we observe that a quasi-isometric
reduction from $\alpha$ to $\beta$ can be obtained from
an infinite branch of the Kleene tree.

\smallskip
\begin{claim}\label{clm:main.clm11}
$\alpha \leq_{qi} \beta$.
\end{claim}

\smallskip
\begin{claimproof}
Fix an infinite branch $\mathcal{B}(1)\mathcal{B}(2)\cdots$ 
of the Kleene tree. For each $n \in \bbbn$, we describe a 
mapping from $\theta_n$ to $\zeta_n$, using the structures
depicted in Equation (\ref{eqn:theta.zeta}) as a reference. 
This will give a quasi-isometric reduction from $\alpha$ to 
$\beta$.
The mappings to be defined are strictly increasing.

First,
since $\theta_1 = \zeta_1 = \lambda_{(n, 0)}$,
we can map $\theta_1$ to $\zeta_1$ in a strictly increasing fashion. 
The lead in the next segment is 0.
For $n \geq 2$, map each join segment $v_{n, i}$ 
of $\theta_n$ to the corresponding join segment $v'_{n, i}$ of 
$\zeta_n$, shifted by the lead at the current step.
That is, suppose the lead is $\ell_1$ and
the number of blocks in $\alpha$ before this segment is $B^n_{2i - 1}$.
Then, for $1 \leq i \leq 3nB^n_{2i - 1} - \ell_1$,
map the $i$-th block of the join segment of $\theta_n$ into
the $(i + \ell_1)$-th block of the corresponding join segment of $\zeta_n$.
Map the last $\ell_1$ blocks $\lambda_{(n, 0)}$ of the join segment into
the first $\ell_1$ blocks of the following segment in $\zeta_n$---which
must be $\lambda_{(n, 0)}$ or $\lambda_{(n, 1)}$ blocks.
Note that a $\lambda_{(n, 0)}$ block can be mapped into
a $\lambda_{(n, 1)}$ block.
\[
\begin{tikzpicture}
\node at (0,0) (lamn0s) {$0^{n}$};
\node at (0.7,0) (lamn1s) {$1^{n}$};
\node at (0,-1) (lamn101) {$0^{\lfloor \frac{n+1}{2} \rfloor}$}; 
\node at (0.7,-1) (lamn11) {$1$};
\node at (1.4,-1) (lamn102) {$0^{\lceil \frac{n+1}{2} \rceil}$};
\node at (2.1,-1) (lamn11s) {$1^{n}$};
\
\draw[->,thick]
(lamn0s) -- (lamn101);
\draw[->,thick]
(lamn0s) -- (lamn102);
\draw[->,thick]
(lamn1s) -- (lamn11s);
\end{tikzpicture}
\]
The lead in the next segment is $\ell_1$.

Next we describe the mapping for scaling part $s_{n, i}$ with lead $\ell_2$.
For each $k \leq n-1$,
map the first $nB^n_{2i}-\ell_2$ blocks $\lam_{(n,k)}$ of $s_{n, i}$
to the corresponding last 
$nB^n_{2i}-\ell_2$ blocks $\lam_{(n,k)}$ of $s'_{n, i}$.
Then, for each $k \leq n-2$, map the last $\ell_2$ 
blocks $\lam_{(n,k)}$ of $s_{n, i}$ to the first $\ell_2$ blocks $\lam_{(n,k+1)}$ of $s'_{n, i}$.
Observe that each block $\lambda_{(n, k)}$ can be mapped to
a block $\lambda_{(n, k + 1)}$ in a strictly increasing manner.
Further, map each of the last $\ell_2$ blocks $\lam_{(n, n - 1)}$
to exactly two blocks $\lam_{(n, 0)}$ of $s'_{n, i}$.
Map the first $2nB^n_{2i} - 2\ell_2$ blocks $\lam_{(n, 0)}$ to
the remaining $2nB^n_{2i} - 2\ell_2$ blocks $\lam_{(n, 0)}$ of $s'_{n, i}$.
Map the last $2\ell_2$ blocks $\lam_{(n, 0)}$ to
the first $2\ell_2$ blocks of the following join segment in $\zeta_n$.
The lead of the next join segment is $2\ell_2$.

For the branching segment, suppose that the current lead is
$\ell_3$. If $\mathcal{B}(n-1) = 1$, 
map the $(2nB^n_6 - \ell_3)$-th $\lam_{(n,0)}$ block to 
the concatenation $\lam_{(n,0)}\lam'_n$ of two blocks in $t'_n v_{n, 4}$.
Otherwise, map the $(2nB^n_6 - \ell_3 + 1)$-st $\lam_{(n,0)}$
block to the $\lam'_n$ block in $t'_n$.
Map the rest of the $\lam_{(n,0)}$ blocks 
such that $f$ is strictly increasing.
Then, the lead of the next join segment is $\ell_3 + \mathcal{B}(n-1)$.

For the selection segment, suppose that the current 
lead is $\ell_4$. Map the $\lam_{(n,1)}$ block to the
$(\ell_4+1)$-st block. By induction, $\ell_4 \in S_n$
and so the $(\ell_4+1)$-st block in the selection segment
of $\zeta_n$ is $\lam_{(n,1)}$.
Recall that a $\lam_{(n,0)}$
block can be mapped in a strictly increasing fashion to a $\lam_{(n,1)}$
block.
Thus, the remaining $\lam_{(n,0)}$
blocks can be mapped to the subsequent $\lam_{(n,0)}$
or $\lam_{(n,1)}$ blocks in a strictly increasing manner.   
\end{claimproof}

From Claims \ref{clm:main.clm10} and \ref{clm:main.clm11}, we conclude that
$\alpha \leq_{qi} \beta$ but $\alpha \not\leq_{mqi} \beta$.
Hence, mqi-reducibility is strictly a stronger notion 
than general quasi-isometric reducibility.
\end{proof}

\section{Automatic Quasi-Isometric Reductions}

Automatic structures were introduced independently 
by Hodgson in 1983~\cite{Hodgson83} and
by Khoussainov and Nerode in 1995~\cite{Khoussainov95}.
Since then, automatic structures have been well studied,
with many surveys written on this
topic~\cite{Gradel20,Khoussainov07,Rubin08,Stephan15}.

The results in the previous section focus mainly on 
quasi-isometric reductions which are recursive.
In this section, we extend the previous results to the area of automata theory
by studying quasi-isometric reductions which are automatic.
We first define formally the notion of automatic quasi-isometric reductions,
by replacing the recursive infinite strings
in Definition~\ref{defn:C.reduction}
with an isomorphic automatic structure.

We let the domain be some regular set $D \subseteq \Sigma^*$
for some finite alphabet $\Sigma$
and let $<_{llex}$ be the length lexicographic ordering.
We define the successor function $succ : D \to D$ such that
$succ(x) = \min_{llex}\{y \in D : x <_{llex} y\}$.
Addition by a constant is then defined using the successor function
such that for any $x \in D$ and $k \in \bbbn_0$, $x \oplus k \coloneq succ^k(x)$.
Similarly, $x \ominus k \coloneq succ^{-k}(x)$ is defined
whenever there are at least $k$ elements of $D$
that are length-lexicographically less than $x$.
Then, given a constant $k \in \bbbn_0$,
the successor function and addition by constant $k$
can be performed automatically.

From the previous paragraph, $(D, <_{llex}, succ)$ is
an automatic structure that is isomorphic to $(\bbbn, <, succ)$,
where $succ$ denotes the successor function in the respective domain.
Then, quasi-isometry results proven earlier for infinite strings
also apply to the corresponding automatic structures
by replacing $(\bbbn, <, succ)$ accordingly with
$(D, <_{llex}, succ)$.
In particular, Proposition~\ref{prop:quasi.iso.alt.def} holds
when we replace $\bbbn, <$ and $succ$ with
$D, <_{llex}$ and $succ$ respectively.
Hence, we can redefine the notion of quasi-isometric reducibility
between two automatic colourings of an automatic domain as follows.

\begin{definition}\label{defn:auto.C.reduction}
Let $C \in \bbbn$, $D \subseteq \Sigma^*$ be regular and
$\alpha, \beta$ be automatic colourings of $D$, that is,
automatic functions from $D$ to some finite sets.
A {\em $C$-quasi-isometric reduction} from $\alpha$ to $\beta$ is
a colour-preserving function
$f : D \to D$ such that for all $x, y \in D$,
\begin{enumerate}[(a)]
\item $f(\min_{llex} D) \leq_{llex} \min_{llex} D \oplus C$ and
$f(x) \ominus C \leq_{llex} f(x \oplus 1) \leq_{llex} f(x) \oplus C$; and
\item $x \oplus C <_{llex} y \Rightarrow f(x) <_{llex} f(y)$.
\end{enumerate}
\end{definition}

Where appropriate, we may drop the constant $C$ and simply call $f$
a quasi-isometric reduction, or a quasi-isometry, from $\alpha$ to $\beta$.
We are particularly interested in the case where $f$ is automatic.

\begin{definition}[Automatic Quasi-Isometric Reducibility]
Let $D \subseteq \Sigma^*$ be a regular set.
An automatic colouring $\alpha$ of $D$ is
{\em automatically quasi-isometrically reducible},
or {\em aqi-reducible}, to another automatic colouring $\beta$ of $D$
iff there exists a quasi-isometric reduction $f$ from $\alpha$ to $\beta$
such that $f$ is automatic.
\end{definition}
% Karen: Do we define many-one, one-one, permutation?

Note that Definition~\ref{defn:auto.C.reduction} uses
the alternate definition of quasi-isometry
given in Definition~\ref{defn:C.reduction},
which is a simpler but equivalent version of
Definition~\ref{defn:quasiisometry}.
The advantage of using this simpler definition is that
it allows automatic quasi-isometry to be defined
for a broader range of domains.
For example, the underlying metric space $(D, d_{llex})$
of the automatic structure $(D, <_{llex}, succ)$
is not always automatic as
subtraction is not automatic for many regular languages.

For our results, we define the following property about domain $D$.

\begin{definition}[\cite{Incitti01}]
Given an infinite regular set $D \subseteq \Sigma^*$,
we define {\em the growth of $D$} as
the function $growth_D : \bbbn_0 \to \bbbn_0$ such that
$growth_D(n) = |\{\sigma \in D : |\sigma| \leq n\}|$.
We say that $D$ has:
\begin{itemize}
\item {\em linear growth} if $growth_D(n) = \Theta(n)$;
\item {\em superlinear growth} if $growth_D(n) = \omega(n)$;
\item {\em polynomial growth} if $growth_D(n) = \Theta(n^c)$
for some $c \in \bbbr$ such that $c \geq 0$;
\item {\em exponential growth} if $growth_D(n) = \Omega(c^n)$
for some $c \in \bbbr$ such that $c > 1$.
\end{itemize}
\end{definition}

Results in automata theory show that if $D$ is regular,
then the growth is either polynomial or exponential, with nothing in between
\cite{Bridson02,Ibarra86,Incitti01,Latteux84,Raz97,Szilard92,Trofimov81}.
Furthermore, if the growth is polynomial, 
then it must be $\Theta(n^c)$ for some $c \in \bbbn$~\cite{Szilard92}.

Besides the growth, which counts the number of strings {\em up to} length $n$,
another useful property of a regular language is
the number of strings at {\em exactly} length $n$.

\begin{lemma}[\protect{\cite[Lemma 1]{Szilard92}}]
\label{lem:growth.lower.bound}
Suppose that $D$ contains $uv_1^*w_1 \ldots v_c^*w_c$ as a subset,
where for each $1 \leq i \leq c$, $v_i$ is a non-empty string and
the first character after $v_i^*$, if it exists, is not the same as
the first character of $v_i$.
If $n = |uw_1 \ldots w_c| + k|v_1|\ldots|v_c|$ for some $k \in \bbbn$,
then the number of strings in $D$ of length $n$ is $\Omega(n^{c-1})$.
\end{lemma}

\begin{lemma}[\protect{\cite[Lemma 4]{Szilard92}}]
\label{lem:growth.upper.bound}
Suppose that a language $D$ is
a finite union of the regular expressions of the form
$uv_1^*w_1 \ldots v_j^*w_j$ where $j \leq c$.
Then, the number of strings in $D$ of length $n$ is $O(n^{c-1})$.
\end{lemma}

We will now present our main results for this section,
which shows that whether automatic quasi-isometric reducibility exists
between two automatic colourings depends on the growth of the domain
(among other factors).

For the next theorem, we define $\beta: D \to \Gamma$ to be
{\em eventually periodic} if and only if
there are $k \in \bbbn$ and $h \in D$ such that for all $x \geq_{llex} h$,
$\beta(x \oplus k) = \beta(x)$.

\begin{theorem}\label{thm:auto.eventually.periodic}
Let $D \subseteq \Sigma^*$ be an infinite regular set,
and let $\alpha, \beta$ be automatic colourings of $D$
such that $\beta$ is eventually periodic.
The following are equivalent:
\begin{enumerate}[(a)]
\item Every colour in $\alpha$ also occurs in $\beta$ and every colour which occurs infinitely often in $\alpha$ also occurs infinitely often in $\beta$;
\item $\alpha$ is aqi-reducible to $\beta$;
\item $\alpha$ is quasi-isometrically reducible to $\beta$.
\end{enumerate}
\end{theorem}

\begin{proof}
$(a) \Rightarrow (b)$: 
We define the quasi-isometric reduction $f$ from $\alpha$ to $\beta$ such that:
\[
f(x) =
\begin{cases}
\min_{llex}\{y \in D: \alpha(x) = \beta(y)\}
& \text{if } \alpha(x) \text{ occurs only finitely often in } \alpha; \\
\min_{llex}\{y \in D: \alpha(x) = \beta(y) \text{ and } x <_{llex} y\}
& \text{if }\alpha(x) \text{ occurs infinitely often in } \alpha. \\
\end{cases}
\]
Note that $f$ is automatic since it is first-order defined
from automatic parameters $D$, $\alpha$, $\beta$ and $<_{llex}$.
Furthermore, $f$ is clearly colour-preserving.
It remains to show that $f$ satisfies Conditions (a) and (b) of
Definition~\ref{defn:auto.C.reduction}.

Let $x_0 = \min_{llex} D$.
Since $\beta$ is eventually periodic, there are $k, h \in \bbbn$ such that
for all $y \geq_{llex} x_0 \oplus h$, $\beta(y \oplus k) = \beta(y)$.
Furthermore, there exists $\ell \geq h$ such that for all $x \geq x_0 \oplus \ell$,
$\alpha(x)$ occurs infinitely often in $\alpha$.
Choose $C' = \max(\{t \in \bbbn_0 : x_0 \oplus t = f(x)$ for some
$x <_{llex} x_0 \oplus \ell\} \cup \{k, \ell\})$.

We first show that for any $x \in D$,
$x \leq_{llex} f(x) \oplus C'$ and $f(x) \leq_{lleq} x \oplus C'$.
If $x <_{llex} x_0 \oplus \ell$, then
$x \leq_{llex} f(x) \oplus C'$ and $f(x) \leq_{lleq} x \oplus C'$ by definition of $C'$.
Otherwise, then $\alpha(x)$ occurs infinitely often and
for all $x' \geq x$, $\beta(x' \oplus k) = \beta(x')$.
Hence, $x <_{llex} f(x) \leq_{llex} x \oplus k$.
Then, $x \leq_{llex} f(x) \oplus C'$ and $f(x) \leq_{llex} x \oplus C'$.

We now show that $f$ is a $C$-quasi-isometric reduction with $C = 2C' + 1$.
The first half of Condition (a) of Definition~\ref{defn:auto.C.reduction}
follows directly from our claim that $f(x) \leq_{llex} x \oplus C'$ for all $x \in D$.
To prove the second half, we have
\[
f(x) \leq_{llex} x \oplus C'
<_{llex} x \oplus 1 \oplus C'
\leq_{llex} f(x \oplus 1) \oplus 2C',
\]
and so $f(x \oplus 1) \geq_{llex} f(x) \ominus C$.
And similarly,
\[
f(x \oplus 1) \leq_{llex} x \oplus 1 \oplus C'
<_{llex} f(x) \oplus 2C' \oplus 1
= f(x) \oplus C.
\]
To prove Condition (b) of Definition~\ref{defn:auto.C.reduction},
suppose that $x \oplus C <_{llex} y$.
Note that
\[
f(x) \oplus C' \leq_{llex} x \oplus 2C' <_{llex} y \leq_{llex} f(y) \oplus C'.
\]
Then, $f(x) <_{llex} f(y)$.

$(b) \Rightarrow (c)$:
This follows from the definition of automatic quasi-isometric reducibility.

$(c) \Rightarrow (a)$:
Let $f$ be a $C$-quasi-isometric reduction from $\alpha$ to $\beta$.
Since $f$ is colour-preserving, then clearly,
every colour in $\alpha$ also occurs in $\beta$.
Furthermore, by Lemma~\ref{lem:collision}, $f$ is finite-to-one.
Hence, every colour which occurs infinitely often in $\alpha$
also occurs infinitely often in $\beta$.
\end{proof}

We next show that eventual periodicity is related to linear growth by the following proposition.

\begin{proposition}
Suppose an infinite regular set $D \subseteq \Sigma^*$ has linear growth.
Then, any automatic colouring $\beta: D \to \Gamma$ is eventually periodic.
\end{proposition}

\begin{proof}
Since $D$ has linear growth, then by Lemmas~\ref{lem:growth.lower.bound} and \ref{lem:growth.upper.bound}, it must be of the form $D = \bigcup_{i=1}^t u_iv_i^*w_i \cup D_{fin}$ where $D_{fin}$ is a finite set, $u_i, w_i \in \Sigma^*$ and $v_i \in \Sigma^+$.
Without loss of generality, we can assume that $|u_1| = \ldots = |u_t|$ and $|v_1| = \ldots = |v_t|$, since $v_i$ can be replaced by $v_i^{|v_1|\ldots|v_{i-1}||v_{i+1}|\ldots|v_t|}$ and rotated as needed so that $|u_1| = \ldots = |u_t|$.

Now consider the set $D_c$ of all the strings $x \in D$ with $\beta(x) = c$.
Since $\beta$ is regular, then so is $D_c$.
Thus, $D_c$ satisfies the pumping lemma for any large enough pumping constant $k$.
We consider the following version of the pumping lemma: for any string $x \in D_c$ of length at least $3k$, there is a representation $x = uvw$ such that $|u|, |w| \geq k$, $0 < |v| \leq k$ and $uv^*w \subseteq D_c$.
We choose a pumping constant $k$ such that $k \geq |u_iv_i|, |v_iw_i|$ for any $1 \leq i \leq t$.
Then, any $x$ of length at least $3k$ must be equal to $u_iv_i^jw_i$ for some $1 \leq i \leq t$ and $j \in \bbbn$.
Furthermore, as $D_c \subseteq D$, the representation $x = uvw$ must satisfy that $u_i$ is a prefix of $u$, $w_i$ is a suffix of $w$ and $v = v_i^{k'}$ for some $k' \leq k$.
Hence, the string $u_iv_i^{j+k!}w_i$ must be a pumped word of $u_iv_i^jw_i$, and so $u_iv_i^{j+k!}w_i \in D_c$ as well.

Next, we claim that there is a constant $m$ depending only on $k$ such that $u_iv_i^{j+k!}w_i = u_iv_i^jw_i \oplus m$ for all $i$ and large enough $j$.
First, observe that if $u_iv_i^jw_i$ is the $h$-th string in $D$ of its length, then $u_iv_i^{j+k!}w_i$ is also the $h$-th string in $D$ of its length.
Furthermore, the number of strings in $D$ of length $|u_iv_i^jw_i|$ to $|u_iv_i^{j+k!}w_i|-1$ is a constant $m$ depending only on $k$.
Hence, it follows that $u_iv_i^{j+k!}w_i = u_iv_i^jw_i \oplus m$ for all $i$ and large enough $j$.

Therefore, we can conclude that for large enough $x$, if $x \in D_c$, then $x \oplus m$ is also in $D_c$.
In other words, if $\beta(x) = c$, then $\beta(x \oplus m) = \beta(x) = c$. Note that this is true for all colours $c \in \Gamma$. Hence, for each colour $c \in \Gamma$, there is a string $x_c \in D$ and a constant $m_c \in \bbbn$ such that for all $x \geq_{llex} x_c$ with $\beta(x) = c$, we have that $\beta(x) = \beta(x \oplus m_c) = c$.
Then, by letting $M = \prod_{c \in \Gamma} m_c$ and $x_0 = \max_{llex} \{x_c : c \in \Gamma\}$, we have that for all $x \geq_{llex} x_0$, $\beta(x) = \beta(x \oplus M)$.
So, $\beta$ is eventually periodic.
\end{proof}

Hence, we have the following results.

\begin{corollary}\label{cor:auto.linear.growth}
Suppose an infinite regular set $D \subseteq \Sigma^*$ has linear growth,
and let $\alpha, \beta$ be automatic colourings of $D$.
The following are equivalent:
\begin{enumerate}[(a)]
\item Every colour in $\alpha$ also occurs in $\beta$ and every colour which occurs infinitely often in $\alpha$ also occurs infinitely often in $\beta$;
\item $\alpha$ is aqi-reducible to $\beta$;
\item $\alpha$ is quasi-isometrically reducible to $\beta$.
\end{enumerate}
\end{corollary}

The above result does not hold if $D$ has superlinear growth.
We first show a counterexample where Condition (a) holds,
but Condition (c) does not hold, and so Condition (b) does not hold.

\begin{theorem}
Suppose an infinite regular set $D \subseteq \Sigma^*$ has superlinear growth.
There exist automatic colourings $\alpha, \beta$ of $D$ such that the following statements are true:
\begin{enumerate}[(a)]
\item Every colour in $\alpha$ also occurs in $\beta$ and every colour which occurs infinitely often in $\alpha$ also occurs infinitely often in $\beta$;
\item $\alpha$ is {\em not} quasi-isometrically reducible to $\beta$.
\end{enumerate}
\end{theorem}

\begin{proof}
We define $\alpha$ and $\beta$ as follows.
Let $\alpha(x) = 1$ if
$x$ is the length-lexicographically minimum string of its length
and $\alpha(x) = 0$ otherwise.
Let $\beta(x) = 1 - \alpha(x)$.
The range of $\alpha$ and $\beta$ is the same, which is $\{0, 1\}$,
and both $0$ and $1$ occur infinitely often in both $\alpha$ and $\beta$.
Hence, $\alpha$ and $\beta$ satisfy statement (a).

We now show that $\alpha$ is not quasi-isometrically reducible to $\beta$
using Corollary~\ref{cor:nary.seq}.
We first claim that there exists $K \in \bbbn$ such that
for any $x \in  D$, there is some $x' \in D$ such that
$x \leq_{llex} x' <_{llex} x \oplus K$ and $\alpha(x') = 0$.
Let $S = \alpha^{-1}(0)$.
Since $\alpha$ is automatic, then $S$ must be regular.
Furthermore, since $D$ has superlinear growth, then by definition of $\alpha$, 
$S$ must be infinite.
So, by the pumping lemma,
$S$ must have a subset of the form $uv^*w$
for some strings $u, w \in \Sigma^*$ and $v \in \Sigma^+$.
Let $K = |uvw| + 2$.
For each $x \in D$, there is an element $x'' \in S$ such that
$|x| < |x''| \leq |x| + |uvw|$.
Then, $x'' \geq_{llex} x$.
If we also have that $x'' <_{llex} x \oplus K$, then we are done.
Otherwise, $x'' \geq_{llex} x \oplus K$ and so $|x| \leq |x \oplus K \ominus 1| \leq |x''| \leq |x| + K - 2$.
Note that for any length $n$,
there is at most one string $y \in D$ of length $n$
such that $\alpha(y) \neq 0$.
Since $S = \alpha^{-1}(0)$,
there are at most $K - 1$ elements of $D \setminus S$
of length $|x|$ to $|x| + K - 2$.
Hence, there are at most $K - 1$ elements of $D \setminus S$
which are length-lexicographically between $x$ and $x \oplus K \ominus 1$ inclusive.
Since there are $K$ elements in $D$ which are
length-lexicographically between $x$ and $x \oplus K \ominus 1$ inclusive,
one of them must be in $S$.
Therefore, for any $x \in  D$, there is some $x' \in S$ such that
$x \leq_{llex} x' <_{llex} x \oplus K$.
Moreover, by definition of $S$, $\alpha(x') = 0$.

Now suppose that there exists a $C$-quasi-isometry $f$ from $\alpha$ to $\beta$
for some $C$.
Since the quasi-isometry defined by
the automatic structure $(D, <_{llex}, succ)$ is isomorphic to
the quasi-isometry between strings defined by $(\bbbn, <, succ)$,
Corollary~\ref{cor:nary.seq} implies that for any $x \in  D$,
there is some $y \leq_{llex} y' < y \oplus KC$ such that $\beta(y') = 0$.
On the other hand, this implies that
there are at most $KC$ strings of each length,
which implies that there are at most $KCn$ strings of length up to $n-1$.
This contradicts the assumption that $D$ has superlinear growth.
Hence, $\alpha$ is not quasi-isometrically reducible to $\beta$.
\end{proof}

Furthermore, if the growth of $D$ is also polynomial (on top of being superlinear),
we can show that 
there exist automatic colourings $\alpha, \beta$ such that
$\alpha$ is quasi-isometrically reducible to $\beta$ but not automatically.
That is, we can separate the notion of
quasi-isometric reducibility from its automatic counterpart.

\begin{theorem}
Suppose an infinite regular set $D \subseteq \Sigma^*$ has superlinear but polynomial growth.
There exist automatic colourings $\alpha, \beta$ of $D$ such that the following statements are true:
\begin{enumerate}[(a)]
\item $\alpha$ is {\em not} aqi-reducible to $\beta$;
\item $\alpha$ is quasi-isometrically reducible to $\beta$.
\end{enumerate}
\end{theorem}

\begin{proof}
Suppose that $D$ has polynomial growth with degree $c \geq 2$.
Then, by Lemmas~\ref{lem:growth.lower.bound} and \ref{lem:growth.upper.bound},
$D$ must have a subset of the form
$u v_1^* w_1 v_2^* w_2 \ldots v_c^* w_c$
where for each $1 \leq i \leq c$,
$v_i$ is a non-empty string
and the first character after $v_i^*$, if it exists,
is not the same as the first character of $v_i$.

Define $t = v_1^{|v_2| \cdot \ldots \cdot |v_c|}$.
Then, $u t^* w_1 \ldots w_c$ is a subset of $D$
and $|t| = |v_1| \cdot \ldots \cdot |v_c| > 0$.
Let $w = w_1 \ldots w_c$ and
$\alpha, \beta : D \to \{0, 1\}$ be defined such that:
\begin{itemize}
\item $\alpha(x) = 1$ if and only if $x \in u (t^2)^* w$;
\item $\beta(x) = 1$ if and only if $x \in u (t^4)^* w$.
\end{itemize}
We first prove statement (a) and show that
$\alpha$ is not aqi-reducible to $\beta$.
Suppose that $f$ is an automatic $C$-quasi-isometry from $\alpha$ to $\beta$.
By Lemma~\ref{lem:growth.lower.bound},
there are $\Omega(n^{c-1})$ strings in $D$ of length $n = |uw| + (2i+1)|t|$
for any $i \in \bbbn_0$.
Since $|ut^{2i}w| < n < |ut^{2(i+1)}w|$,
there are $\Omega(n^{c-1})$ strings in $D$ which are
length-lexicographically between $ut^{2i}w$ and $ut^{2(i+1)}w$.
Note that $c \geq 2$, and so there exists some $i_0$ such that
for all $i \geq i_0$, $u t^{2i} w \oplus C <_{llex} u t^{2(i+1)} w$.
Then, $f(u t^{2i} w) <_{llex} f(u t^{2(i+1)} w)$.

Define $g : \bbbn_0 \to \bbbn_0$ such that
$f(u t^{2i} w) = u t^{4g(i)} w$ for any $i \in \bbbn_0$.
Such $g$ is well-defined by the definitions of $\alpha$ and $\beta$,
as well as the colour-preserving property of $f$.
Furthermore, since $f$ is automatic, then so is $g$.
From the previous paragraph,
we know that for all $i \geq i_0$, $g(i+1) \geq g(i) + 1$.
Then, for any $i \in \bbbn_0$, $g(i_0+i) \geq g(i_0) + i$.
Therefore, for any $x \in \bbbn_0$,
\[
f(u t^{2(i_0+i)} w) = u t^{4(g(i_0+i))} w \geq_{llex} u t^{4(g(i_0)+i)} w.
\]
This contradicts the assumption that $f$ is automatic.
Hence, an automatic quasi-isometry from $\alpha$ to $\beta$ cannot exist.

We now prove statement (b) and show that
$\alpha$ is quasi-isometrically reducible to $\beta$.

As shown earlier, Lemma~\ref{lem:growth.lower.bound} implies that
there are $\Omega(i^{c-1})$ strings in $D$ which are length-lexicographically
between $ut^{2i}w$ and $ut^{2(i+1)}w$ for any $i \in \bbbn_0$.
Furthermore, by Lemma~\ref{lem:growth.upper.bound},
there are at most $(2|t|+1)O(i^{c-1}) = O(i^{c-1})$ strings in $D$
which are length-lexicographically between $ut^{2i}w$ and $ut^{2(i+1)}w$.
Hence, there exist some $d_1, d_2 \in \bbbr^+$ such that
for all large enough $i$, we have
\begin{equation}\label{eqn:auto.poly.alpha.distribution}
d_1 i^{c-1} \leq |\{x \in D: ut^{2i}w <_{llex} x <_{llex} ut^{2(i+1)}w\}|
\leq d_2 i^{c-1}.
\end{equation}
By a similar argument, there exist some $d_3, d_4 \in \bbbr$ such that
for all large enough $i$, we have
\begin{equation}\label{eqn:auto.poly.beta.distribution}
d_3 i^{c-1} \leq |\{y \in D: ut^{4i}w <_{llex} y <_{llex} ut^{4(i+1)}w\}|
\leq d_4 i^{c-1}.
\end{equation}
Then, there is some $i_0 \in \bbbn$ such that
$(\ref{eqn:auto.poly.alpha.distribution})$ and
$(\ref{eqn:auto.poly.beta.distribution})$ are both true
for all $i \geq i_0$.

We define $f: D \to D$ as follows. For any $i \geq i_0$:
\begin{itemize}
\item $f(ut^{2i}w) = ut^{4i}w$;
\item $f(ut^{2i}w \oplus 1), f(ut^{2i}w \oplus 2), \ldots, f(ut^{2(i+1)}w \ominus 1)$
are distributed evenly between $ut^{4i}w$ and $ut^{4(i+1)}w$
in a non-decreasing manner.
\end{itemize}
Furthermore, for any $x <_{llex} ut^{2i_0}w$,
$f(x)$ is the length-lexicographically minimum string in $D$
such that $\beta(x) = \alpha(x)$.

Let $C = \max\{\lceil d_4/d_1 \rceil, \lceil d_2/d_3 \rceil, j_0\}$
where $ut^{4i_0}w = \min_{llex} D \oplus j_0$.
We show that $f$ is a $C$-quasi-isometry.
First, we show that for any $x <_{llex} ut^{2i_0}w$, $f(x) <_{llex} ut^{4i_0}w$.
Note that $i_0 \geq 1$.
Then, $uw <_{llex} ut^{3i_0}w <_{llex} ut^{4i_0}w$.
Moreover, $\beta(uw) = 1$ and $\beta(ut^{3i_0}w) = 0$.
So, for any $x <_{llex} ut^{2i_0}w$,
the length-lexicographically minimum $z \in D$ such that $\beta(z) = \alpha(x)$
satisfies $z <_{llex} ut^{4i_0}w$.
Hence, $f(x) = z <_{llex} ut^{4i_0}w$.

Then, by definition of $C$,
$f(\min_{llex} D) <_{llex} ut^{4i_0}w \leq_{llex} \min_{llex} D \oplus C$
since $\min_{llex} D <_{llex} ut^{2i_0}w$.
Similarly, if $x <_{llex} ut^{2i_0}w$, then by definition of $C$,
$f(x) \ominus C \leq_{llex} f(x \oplus 1) \leq_{llex} f(x) \oplus C$.
Now suppose that $x \geq_{llex} ut^{2i_0}w$.
Clearly, $f(x \oplus 1) \geq_{llex} f(x) \geq_{llex} f(x) \ominus C$.
Let $i$ be the largest $i'$ such that $ut^{2i'}w \leq_{llex} x$.
Note that $i \geq i_0$.
So, from (\ref{eqn:auto.poly.alpha.distribution}) and
(\ref{eqn:auto.poly.beta.distribution}), the ratio
\[
\frac{|\{y \in D: ut^{4i}w <_{llex} y <_{llex} ut^{4(i+1)}w\}|}
{|\{x \in D: ut^{2i}w <_{llex} x <_{llex} ut^{2(i+1)}w\}|} \leq \frac{d_4}{d_1}.
\]
Hence, by definition of $f$,
$f(x \oplus 1) \leq_{llex} f(x) \oplus \lceil d_4/d_1 \rceil
\leq_{llex} f(x) \oplus C$.
Therefore, $f$ satisfies Condition (a) of
Definition~\ref{defn:auto.C.reduction}.

We now show that $f$ satisfies Condition (b) of
Definition~\ref{defn:auto.C.reduction}.
Suppose that $x \oplus C <_{llex} y$.
By definition of $C$, $y \geq_{llex} ut^{2i_0}w$.
First, consider the case where $x <_{llex} ut^{2i_0}w$,
so that $f(x) <_{llex} ut^{4i_0}w$.
Observe that $f$ is non-decreasing from $ut^{2i_0}w$ onwards,
and so $f(y) \geq_{llex} f(ut^{2i_0}w) = ut^{4i_0}w$.
Then, we have $f(x) <_{llex} ut^{4i_0}w \leq_{llex} f(y)$.

Now we consider the case that $x \geq_{llex} ut^{2i_0}w$.
Let $i$ be the largest $i'$ such that $ut^{2i'}w \leq_{llex} x$.
Clearly, $i \geq i_0$.
If $x = ut^{2i}w$ or $ut^{2(i+1)}w \leq_{llex} y$,
then clearly, $f(x) <_{llex} y$.
So, we can assume that $ut^{2i}w <_{llex} x <_{llex} y <_{llex} ut^{2(i+1)}w$.
By (\ref{eqn:auto.poly.alpha.distribution}) and
(\ref{eqn:auto.poly.beta.distribution}), the ratio
\[
\frac{|\{x \in D: ut^{2i}w <_{llex} x <_{llex} ut^{2(i+1)}w\}|}
{|\{y \in D: ut^{4i}w <_{llex} y <_{llex} ut^{4(i+1)}w\}|} \leq \frac{d_2}{d_3}.
\]
Hence, there are at most $\lceil d_2/d_3 \rceil \leq C$ consecutive strings
$x', x' \oplus 1, \ldots, x' \oplus j$ between $ut^{2i}w$ and $ut^{2(i+1)}w$
which are mapped to the same $y'$ between $ut^{4i}w$ and $ut^{4(i+1)}w$.
So, since $y >_{llex} x \oplus C$, then $f(y) >_{llex} f(x)$.
\end{proof}

On the other hand, if the growth of $D$ is not polynomial, then it must be exponential.
We can separate the notions of quasi-isometric reducibility
from its automatic counterparts for some domains $D$ with exponential growth.

\begin{example}
\label{exmp:auto.exponential.growth}
Let $D \subseteq \{0, 1\}^*$ be defined as follows:
$D$ contains all $\sigma \in \{0, 1\}^*$ of even length,
and all $\tau \in 0^*1^*$ of odd length.
There exist automatic colourings $\alpha, \beta$ of $D$ such that the following statements are true:
\begin{enumerate}[(a)]
\item $\alpha$ is {\em not} aqi-reducible to $\beta$;
\item $\alpha$ is quasi-isometrically reducible to $\beta$.
\end{enumerate}

We define $\alpha, \beta : D \to \{0, 1\}$ as follows:
\begin{itemize}
\item $\alpha(\sigma) = 1$ if and only if $\sigma \in 1^*$ and
$|\sigma|$ is odd.
\item $\beta(\sigma) = 1$ if and only if $\sigma \in 0^*$ and
$|\sigma|$ is odd.
\end{itemize}
We first prove that $\alpha$ is not aqi-reducible to $\beta$.
Suppose, for the sake of contradiction, that
$f$ is a $C$-quasi-isometric reduction from $\alpha$ to $\beta$
which is automatic.
Since $f$ is colour-preserving,
then for any $n \in \bbbn_0$,
$f(1^{2n+1}) = 0^{2m+1}$ for some $m \in \bbbn_0$.
In particular, observe that $f(1^{2n+1}) = 0^{\Theta(n)}$.
Now consider $f(0^{2n+1})$ for large enough $n \in \bbbn_0$.
Note that $0^{2n+1} \oplus (2n+1) = 1^{2n+1}$.
By Condition (b) of Definition~\ref{defn:auto.C.reduction},
$f(0^{2n+1}) \oplus \lfloor (2n+1)/(C+1) \rfloor \leq f(1^{2n+1})$.
On the other hand, by Condition (a) of Definition~\ref{defn:auto.C.reduction},
$f(1^{2n+1}) \leq_{llex} f(0^{2n+1}) \oplus (2n+1)C$.
Hence, $f(0^{2n+1}) = f(1^{2n+1}) \ominus \Theta(n)
= 0^{\Theta(n)} \ominus \Theta(n)$.
In other words, $f(0^{2n+1}) = 1^m 0 \sigma$ for some $\sigma \in \{0, 1\}^*$
such that $|\sigma| = \Theta(\log n)$ and $m + |\sigma| = \Theta(n)$.

Now consider the set $S = f(\{0^{2n+1}:n \in \bbbn_0\})$.
Since $f$ is automatic, then the set $S$ is regular,
and thus satisfies the pumping lemma.
Take some $f(0^{2n'+1})$ for some large enough $n'$.
As shown earlier, $f(0^{2n'+1}) = 1^{m'} 0 \sigma'$
for some  $m' \in \bbbn_0$ and $\sigma' \in \{0, 1\}^*$ such that $|\sigma'| = \Theta(\log n')$.
So, by the pumping lemma,
there are some $u,w \in \{0, 1\}^*$ and $v \in \{0, 1\}^+$
such that $\sigma' = uvw$ and $1^{m'} 0 uv^*w \subseteq S$.
On the other hand, for all large enough $n$,
$f(0^{2n+1})$ has the form $1^m 0 \sigma$
where $|\sigma| = \Theta(\log n)$ and $m + |\sigma| = \Theta(n)$.
This contradicts that $S$ contains all words $1^{m'} 0 uv^iw$
where $i \in \bbbn_0$.
Hence, $\alpha$ is not aqi-reducible to $\beta$.

It remains to show that $\alpha$ is quasi-isometrically reducible to $\beta$.
Let $f: D \to D$ be defined as follows.
Firstly, $f(\epsilon) = f(0) = \epsilon$.
For any $n \in \bbbn_0$:
\begin{itemize}
  \item $f(1^{2n+1}) = 0^{2n+1}$;
  \item $f(1^{2n+1} \oplus 1), f(1^{2n+1} \oplus 2), \ldots,
  f(1^{2(n+1)+1} \ominus 1)$
  are distributed evenly between $0^{2n+1}$ and $0^{2(n+1)+1}$
  in a non-decreasing manner.
\end{itemize}
We shall show that
$f$ is a $1$-quasi-isometric reduction from $\alpha$ to $\beta$.

Note that $\min_{llex} D = \epsilon$.
Clearly, $f(\epsilon) = \epsilon \leq_{llex} \epsilon \oplus 1$.
Furthermore, $f$ is non-decreasing and so
for any $x \in D$, $f(x) \ominus 1 \leq_{llex} f(x) \leq_{llex} f(x \oplus 1)$.
To prove the remaining inequality
in Condition (a) of Definition~\ref{defn:auto.C.reduction},
note that for any $n \in \bbbn_0$,
$0^{2(n+1)+1} = 0^{2n+1} \oplus (2n+2+2^{2n+2})$ and
$1^{2(n+1)+1} = 1^{2n+1} \oplus (2n+4+2^{2n+2})$.
Hence, there are $2n+1+2^{2n+2}$ strings length-lexicographically
between $0^{2n+1}$ and $0^{2(n+1)+1}$,
and $2n+3+2^{2n+2}$ strings between $1^{2n+1}$ and $1^{2(n+1)+1}$.
Since there are fewer strings between $0^{2n+1}$ and $0^{2(n+1)+1}$
than between $1^{2n+1}$ and $1^{2(n+1)+1}$,
by definition of $f$, for any $x \geq_{llex} 1$,
$f(x \oplus 1) = f(x)$ or $f(x) \oplus 1$.
So, $f(x \oplus 1) \leq f(x) \oplus 1$ for all $x \in D$.
Therefore, $f$ satisfies Condition (a) of
Definition~\ref{defn:auto.C.reduction} with $C = 1$.

To prove that $f$ satisfies Condition (b) of
Definition~\ref{defn:auto.C.reduction},
note that for any $n \in \bbbn_0$, $2n+1+2^{2n+2} \geq 5$
and so the number of strings between $1^{2n+1}$ and $1^{2(n+1)+1}$ is less than
two times the number of strings between $0^{2n+1}$ and $0^{2(n+1)+1}$.
Hence, each $y \in D$ is the image of at most $2$ strings $x \in D$.
Since $f$ is non-decreasing, then
$f(x) <_{llex} f(y)$ whenever $x \oplus 1 <_{llex} y$.

Hence, $f$ satisfies Conditions (a) and (b) of
Definition~\ref{defn:auto.C.reduction},
and is a $1$-quasi-isometric reduction from $\alpha$ to $\beta$.
\end{example}

The question remains open on whether there always exist,
for every regular domain $D$ with exponential growth,
automatic colourings $\alpha, \beta$ of $D$ such that
there is a quasi-isometric reduction from $\alpha$ to $\beta$,
but no such reduction is automatic.
In particular, the construction of $\alpha$ and $\beta$
in Example~\ref{exmp:auto.exponential.growth} exploits the fact that
the subset of all strings in $D$ with odd length has quadratic growth,
and $D$ has all the strings with even length (and thus exponential growth).
Since arbitrary domains with exponential growth may not have such nice properties,
our technique may not be easily generalised to all exponential domains.

\section{Conclusions and Future Investigations}

The present paper introduced finer-grained notions of 
quasi-isometries between infinite strings, in particular 
requiring the reductions to be recursive. We showed that
permutation quasi-isometric reductions are provably more
restrictive than one-one quasi-isometric reductions,
which are in turn provably more restrictive than 
many-one quasi-isometric reductions.
One result was that general many-one quasi-isometries
are strictly more powerful than recursive many-one 
quasi-isometries, which answers
Khoussainov and Takisaka's open problem.

This work also presented some results on the structures of the
permutation, one-one and many-one quasi-isometric degrees.
It was shown, for example, that there are two infinite strings
whose many-one quasi-isometric degrees have a unique common
upper bound. It was also proven that the partial order
$\Sigma^{\omega}_{mqi}$ is non-dense with respect to 
pairs of mqi-degrees. We conclude with the simple observation
that the class of mqi-degrees does not form a lattice; in 
particular, the mqi-degrees $[0^{\omega}]_{mqi}$ and
$[1^{\omega}]_{mqi}$ do not have a common lower bound.  

Furthermore, this paper also studied quasi-isometries which are automatic.
The main results show that automatic quasi-isometry can be separated from
general quasi-isometry depending on the growth of the domain.

%%
%% Bibliography
%%

%% Please use bibtex, 

\bibliography{quasi-isom-strings}
\nocite{Kucera85}
\nocite{LiV19}

\appendix
\end{document}